%% file: evm_jan24.tex
\documentclass[english,12pt,draftcls,onecolumn]{IEEEtran}
\ifCLASSINFOpdf
\else
\fi
\usepackage{caption}
\usepackage{epsfig}
\usepackage{amsthm,amssymb,graphicx,graphicx,multirow,amsmath,color,mathtools}
\usepackage{tikz}
\usepackage{pgfplots}

\usepackage{epstopdf}
\usetikzlibrary{patterns}
\usepackage{verbatim}

\include{notation}

\usepackage{subcaption}

\date{}
\begin{document}
%
\title{Error Vector Magnitude Analysis in Generalized Fading with Co-Channel Interference}
%
%
%
 
\author{\IEEEauthorblockN{Sudharsan Parthasarathy${}^{(a)}$, Suman Kumar${}^{(b)}$,  Radha Krishna Ganti${}^{(a)}$,\\ Sheetal Kalyani${}^{(a)}$, K. Giridhar${}^{(a)}$}\\
\IEEEauthorblockA{Department of Electrical Engineering,\\
${}^{(a)}$Indian Institute of Technology Madras, Chennai, India\\
${}^{(b)}$Indian Institute of Technology Ropar, Ropar, India\\
\{sudharsan.p, rganti, skalyani, giri\}@ee.iitm.ac.in, suman@iitrpr.ac.in} 
}

\maketitle


\begin{abstract} 
In this paper, we derive the data-aided Error Vector Magnitude (EVM) in an interference limited system  when both the desired signal and interferers  experience independent and non identically distributed $\kappa$-$\mu$ shadowed fading. Then it is analytically shown that the EVM is equal to the square root of number of interferers when the desired signal and interferers do not experience fading.  Further, EVM is derived in the presence of interference and noise, when the desired signal experiences $\kappa$-$\mu$ shadowed fading and the interferers experience independent and identical Nakagami fading. Moreover, using the properties of the special functions, the derived EVM expressions are also simplified for various special cases. 
\end{abstract}

\section{ \textbf{Introduction}}\label{intro}
The achievable performance of any transmission scheme in a wireless communication system depends on the instantaneous nature of the wireless channel \cite{Hanzo2004}. There are several ways of modelling the wireless channel, where these models depend on the specific nature of the fading considered. The traditional models of the wireless channel include the additive white Gaussian noise (AWGN),  Rayleigh,  Rician, Nakagami-m as well as the Nakagami-q faded channels.  These models have been characterized in several studies  \cite{Mahmoud2009,Aalo}. By contrast, studies of late  are employing  the recently proposed $\eta$-$\mu$ and $\kappa$-$\mu$ fading channels \cite{Yacoub2007}, which provides two significant advantages over   using the traditional channel models. Firstly, the classical channel models are special cases of 
these more sophisticated as well as generalized channel models. Secondly, they fit experimentally measured mobile radio propagation statistics better as compared to the other channel models \cite{Yacoub2007}. Recently, Paris has  proposed $\kappa$-$\mu$ shadowed fading and in \cite{moreno2015kappa} it has been shown that both  $\kappa$-$\mu$ and $\eta$-$\mu$ fading are the special cases of $\kappa$-$\mu$ shadowed fading. 

Bit Error Ratio (BER),  throughput and outage probability are some of the classical metrics used for comparing the performance of various wireless communication systems \cite{Hanzo2004, 7130668}. Error Vector Magnitude (EVM) is an alternate performance metric that is being increasingly employed in the wireless industry as well 
as in the research community  \cite{Agilent}, where the benefits of employing this performance metric are as follows:
  
\begin{enumerate}
\item  The type of degradation as well as its source can be identified when relying on EVM \cite{Agilent}. These impairments include the Local Oscillator's (LO) phase noise, LO frequency error, In-phase Quadrature-phase  (IQ) imbalance,  non-linearity and carrier leakage   \cite{Liu2014}, \cite{Georgiadis2004}. 

\item  EVM is a more convenient symbol-level performance metric than BER for a system that employs adaptive modulation \cite{Shafik2006b}. 

\item It has already become a part of wireless standards.  The Wideband Code Division  Multiple Access (W-CDMA) standard and 
 the  IEEE  802.11 family of  Wireless Local Area Network (WLAN) standards are examples of standards that have incorporated 
EVM-based measurement of the minimum  system performance \cite{Agilent}, \cite{Shafik2006b}.  

\end{enumerate}
The  $\eta$-$\mu$ and $\kappa$-$\mu$ fading channels'  BER,  outage probability as well as  capacity  have been  studied in the presence of interference in \cite{7130668, Paris2013} and references therein. 
On the other hand, the  EVM that may be attained in these wireless channels has not yet been characterized in the presence of interference. However,  the EVM  of generalized fading  channels without considering interference has been characterized in a few studies. 
Specifically, \cite{Gharaibeh2004} constitutes a seminal study in this field that first derived EVM for a digital communication system and related it with SNR. 
The study in  \cite{Shafik2006b} formulates  the EVM attained by a wireless system communicating   over an AWGN channel, 
while the study in \cite{Mahmoud2009} characterized  the  EVM  for   transmission over both 
AWGN as well as Rayleigh channels, when assuming non data-aided reception. Recently, authors of \cite{thomaserror} have derived the  EVM of a single input multiple output system relying on maximal ratio combining  in the presence of generalized fading without considering interference. To the best of our knowledge, none of the prior work in open literature have derived EVM by considering interference. Therefore, in this work we  derive the data-aided EVM  when both user signal and interferers experience generalized fading.

This paper has  the following  novel contributions:  
\begin{enumerate}

\item EVM expression is derived for an interference limited system when both the desired signal and interferers experience independent and non identical distributed (i.n.i.d)  $\kappa$-$\mu$ shadowed fading. The expression is expressed in terms of the Lauricella${}^{'}$s function of the fourth kind, which can be
easily evaluated numerically.
\item It is analytically shown that EVM is equal to the square root of number of interferers when the interferers and desired signal do not experience fading.
\item In a system that experiences both  noise and interference, EVM expression is derived when the desired signal experiences $\kappa$-$\mu$ shadowed fading and the interferers experience independent and  identically distributed (i.i.d) Nakagami fading. 

\item  Using the properties of the special functions, EVM expressions are simplified for various special cases, i.e., when the desired signal and interferers experience $\kappa$-$\mu$, $\eta$-$\mu$, Rayleigh, Rician, Nakagami-m fading channels.

\end{enumerate}

\section{System model}
We consider the following channel model with $L$ interferers, 
\begin{equation}
y(i)=D(i)h+ \sum\limits_{l=1}^L I_l(i)h_l + n(i), \text{ }\forall i=1,2\cdots N \label{conven_channel},
\end{equation}
where $h=a e^{j \theta}$, $h_l=a_l e^{j \theta_l}$.
Here $D(i)$, $I_l(i)$ and $y(i)$ are the desired, interfering and received symbols, respectively in the $i^{th}$ slot.  The imaginary and real components of complex noise term $n(i)$ are independent of each other and are modeled by a zero mean Gaussian distribution of variance $\frac{\sigma^2}{2}$. The fading gain of the desired channel and the interfering channel are  $a$ and $a_l$ respectively, and they are $\kappa$-$\mu$ shadowed distributed. It is also assumed that they  are constant over the block of symbols $1,\cdots, N$ \cite{thomaserror}.  The probability density function (pdf) of the fading power ($a^2$ or $a_l^2$) is given by  \cite{6594884}

\begin{equation}
\label{eqn:kappamupdf}
f_X(x)=\frac{\theta_1^{m-\mu}x^{\mu-1} } {\theta_2^m\Gamma(\mu)}e^{-\frac{x}{\theta_1}} {}_1 F_1 \left(m,\mu,\frac{(\theta_2-\theta_1)x}{\theta_1\theta_2}\right),
\end{equation}
where $\theta_1=\frac{\bar{X}}{\mu(1+\kappa)}$, $\theta_2=\frac{(\mu\kappa+m)\bar{X}}{\mu(1+\kappa)m}$
and $\mu=\frac{E^2\{X\}}{var\{X\}} \frac{1+2\kappa}{(1+\kappa)^2}$. The ratio of the total power in the dominant components to that in the  scattered waves is represented by $\kappa>0$ and the shadowing parameter is denoted by $m$. Here  ${}_1 F_1(.)$ denotes the Kummer confluent hypergeometric function \cite{srivastava1985multiple}.  In \cite{6594884}, \cite{moreno2015kappa} the authors have shown how popular fading distributions such as Rayleigh, Rician, Rician shadowed, Nakagami, $\kappa$-$\mu$, $\eta$-$\mu$, one sided Gaussian, Hoyt etc. can be obtained as special cases of $\kappa$-$\mu$ shadowed fading.

\subsection{EVM Measurement}
 EVM is defined as the root mean squared error between the transmitted symbol and the symbol received (after equalization) \cite{thomaserror}. 
\begin{equation}
\EVM=\E \left(\frac{1}{\sqrt{E_s}}\sqrt{\frac{1}{N}\sum\limits_{i=1}^N \left| \frac{y(i)}{h} -D(i) \right|^2 } \right)
\label{eqn:evmdefn}
\end{equation}
Substituting for $y(i)$ from \eqref{conven_channel} in \eqref{eqn:evmdefn},  EVM is 
\begin{equation}
\textstyle
\small
\E (\sqrt{\frac{\sum\limits_{i=1}^N \sum\limits_{l=1}^L I_l^{*}(i) h_l^{*} (\sum\limits_{j=1}^L I_j(i) h_j +n(i)) + n^{*}(i)( \sum\limits_{j=1}^L I_j(i) h_j +n(i) ) }{N E_s |h|^2}} )
\label{eqn:evmy}
\end{equation}

We  consider the symbols to be symmetric (mean 0) and average energy $E_s$, and symbols from different interferers to be independent, i.e., 
\begin{equation}
  \frac{\sum\limits_{i=1}^N I_l^{*}(i) h^{*}_l I_j(i) h_j }{N} =\begin{cases}
    0, & \text{if $l \neq j$}\\
    |h_l|^2 E_s, & \text{if $l=j$},
  \end{cases}
  \label{eqn:evmenergy}
  \end{equation}
  and 
  \begin{equation}
   \frac{\sum\limits_{i=1}^N I_l^{*}(i) h^{*}_l n(i) }{N} =0.
   \label{eqn:evmmean}
  \end{equation}
  The complex Gaussian noise samples are of zero mean and variance $\sigma^2$. Hence 
  \begin{equation}
   \frac{\sum\limits_{i=1}^N n^*(i) n(i) }{N} =\sigma^2.
   \label{eqn:evmnoise}
  \end{equation}
  Substituting  \eqref{eqn:evmenergy}, \eqref{eqn:evmmean}, \eqref{eqn:evmnoise} in \eqref{eqn:evmy}, and assuming average symbol energy $E_s=1$,

\ifCLASSOPTIONtwocolumn
\begin{equation*}
\EVM= \E \left( \sqrt{\frac{1}{|h|^2}  (\sum\limits_{l=1}^L  |h_l|^2 + \sigma^2 )} \right)
\end{equation*}
\else
\begin{equation*}
\EVM= \E_{h, h_{j \in [1, L]} } \left( \sqrt{\frac{1}{|h|^2}  (\sum\limits_{l=1}^L  |h_l|^2 + \sigma^2 )} \right)
\end{equation*}
\fi
Denoting $g_d=|h|^2$ and $g_I=\sum\limits_{l=1}^L |h_l|^2$, the above expression can be rewritten as
\begin{equation}
\evm = \int\limits_{0}^{\infty} \int\limits_{0}^{\infty} \sqrt{\frac{g_I+\sigma^2}{g_d}} f_{g_d}(g_d) f_{g_I}(g_I) \d g_d \d g_I. 
\label{eqn:evm}
\end{equation}
Note that here $g_d$ and $g_I$ are the desired fading power random variable (RV) and sum of interfering fading powers RV, respectively.

\section{EVM derivation}
In this Section, we derive the EVM expression for the case when only interferers are present and also for the case when both interferers and noise are present. Further, we simplify the EVM expression for various special cases. First, we consider the case when only interferers are present. 

\subsection{Interference limited system }
In this subsection, we derive the EVM expression when both desired and interfering signals experience independent and  non identical $\kappa$-$\mu$ shadowed fading of unit mean power. Let the fading parameters of desired signal power $g_d$ and $l^{th}$ interferer fading power $g_l=|h_l|^2$ be $(\mu, \kappa, m)$ and $(\mu_l, \kappa_l, m_l)$, respectively.  In other words, the pdf of $g_d$ is given as
\begin{equation}
f_{g_d}(x)=\frac{\theta_{1}^{m-\mu}x^{\mu-1} } {\theta_{2}^{m}\Gamma(\mu)}e^{-\frac{x}{\theta_{1}}} {}_1 F_1 \left(m,\mu,\frac{(\theta_{2}-\theta_{1})x}{\theta_{1}\theta_{2}}\right),
\label{eqn:gdpdf}
\end{equation}
where $\theta_{1}=\frac{1}{\mu(1+\kappa)}$, $\theta_{2}=\frac{\mu \kappa +m}{\mu(1+\kappa)m}$.
Similarly the pdf of $g_l$ is 
\begin{equation*}
f_{g_l}(x)=\frac{\theta_{1l}^{m_l-\mu _l}x^{\mu _l-1} } {\theta_{2l}^{m_l}\Gamma(\mu _l)}  \frac{{}_1 F_1 (m_l,\mu _l,\frac{(\theta_{2l}-\theta_{1l}) x}{\theta_{1l}\theta_{2l}})}{e^{\frac{x}{\theta_{1l}}}},
\end{equation*}
where $\theta_{1l}=\frac{1}{\mu_l(1+\kappa_l)}$, $\theta_{2l}=\frac{\mu_l \kappa_l +m_l}{\mu_l(1+\kappa_l)m_l}$. Note that in the EVM expression given in \eqref{eqn:evm}, $g_d$ is $\kappa$-$\mu$ shadowed power RV and $g_I=\sum\limits_{l=1}^L |h_l|^2$, i.e., $g_I$ is the sum of independent and non identical $\kappa$-$\mu$ shadowed power RV. In order to derive EVM, we first need to know the pdf of $g_I$. The pdf of $g_I=\sum\limits_{l=1}^L |h_l|^2$, i.e., the sum of independent and non identical $\kappa$-$\mu$ shadowed power RV is given by \cite{6594884}
\ifCLASSOPTIONtwocolumn
 \begin{equation*}
f_{g_I}(x)=\Phi_2^{2 L} (\mu_1-m_1,\cdots, \mu_L-m_L,m_1,\cdots,m_L; \sum\limits_{l=1}^L \mu_l ; 
\end{equation*}
\begin{equation*}
 -x \mu_1 (1+\kappa_1),\cdots,-x \mu_L (1+\kappa_L), \frac{\mu_1 m_1 (1+\kappa_1) x}{-(\mu_1 \kappa_1 +m_1)},\cdots,  
\end{equation*}
\begin{equation}
\frac{\mu_L m_L (1+\kappa_L) x}{-(\mu_L \kappa_L +m_L)} )  (\prod\limits_{l=1}^L   \frac{\mu_l^{\mu_l} m_l^{m_l} (1+\kappa_l)^{\mu_l} }{(\mu_l \kappa_l +m_l)^{m_l}}   )\frac{x^{(\sum\limits_{l=1}^L \mu_l -1)}}{\Gamma(\sum\limits_{l=1}^L \mu_l)} 
\label{eqn:gIpdf}
\end{equation}
\else
\begin{equation*}
f_{g_I}(x)=\Phi_2^{2 L} (\mu_1-m_1,\cdots, \mu_L-m_L,m_1,\cdots,m_L; \sum\limits_{l=1}^L \mu_l ; -x \mu_1 (1+\kappa_1),\cdots,-x \mu_L (1+\kappa_L), 
\end{equation*}
\begin{equation}
\frac{\mu_1 m_1 (1+\kappa_1) x}{-(\mu_1 \kappa_1 +m_1)},\cdots, \frac{\mu_L m_L (1+\kappa_L) x}{-(\mu_L \kappa_L +m_L)} )   \left(\prod\limits_{l=1}^L   \frac{\mu_l^{\mu_l} m_l^{m_l} (1+\kappa_l)^{\mu_l} }{(\mu_l \kappa_l +m_l)^{m_l}}   \right)\frac{x^{(\sum\limits_{l=1}^L \mu_l -1)}}{\Gamma(\sum\limits_{l=1}^L \mu_l)}  ,  
\label{eqn:gIpdf}
\end{equation}
\fi
 where $\Phi_2^{2L}(.)$ is the confluent Lauricella function \cite{srivastava1985multiple}. We can now substitute the pdf of $g_d$ given in \eqref{eqn:gdpdf} and the pdf of $g_I$ given in \eqref{eqn:gIpdf} into \eqref{eqn:evm} for deriving the EVM expression.

\begin{theorem}
\label{thm1}
The EVM of an interference limited system when both the desired signal and interferers experience i.n.i.d. $\kappa$-$\mu$ shadowed fading is given by 

\ifCLASSOPTIONtwocolumn

\begin{equation*}
 \frac{ {}_2F_1 (\mu - 0.5, m, \mu, \frac{\mu \kappa}{m+ \mu \kappa} ) \Gamma(\mu-0.5)\Gamma(\sum\limits_{l=1}^L \mu_l+0.5 )}{\prod\limits_{l=1}^L  \theta_{1l}^{\mu_l-m_l} \theta_{2l}^{m_l} \theta_1^{\mu-m} \theta_2^{m} \theta_1^{0.5-\mu}  \theta_{11}^{-\sum\limits_{l=1}^L \mu_l-0.5} \Gamma(\mu) \Gamma(\sum\limits_{l=1}^L \mu_l)} \times
\end{equation*}
\begin{equation*}
F_D^{(2L-1)}(\sum\limits_{l=1}^L \mu_l+0.5, \mu_2-m_2,\cdots, \mu_L-m_L,m_1,\cdots, m_L,
\end{equation*}
\begin{equation}
1-\frac{\theta_{11}}{\theta_{12}},\cdots,1-\frac{\theta_{11}}{\theta_{1L}},1-\frac{\theta_{11}}{\theta_{21}},\cdots,1-\frac{\theta_{11}}{\theta_{2L}})
\label{eqn:evmkappamushadowedint}
\end{equation}

\else

\begin{equation*}
\small
\prod\limits_{l=1}^L  \theta_{1l}^{m_l-\mu_l} \theta_{2l}^{-m_l} \theta_1^{m-\mu} \theta_2^{-m} \theta_1^{\mu-0.5}  {}_2F_1 (\mu - 0.5, m, \mu, \frac{\mu \kappa}{m+ \mu \kappa} ) \theta_{11}^{\sum\limits_{l=1}^L \mu_l+0.5} \frac{\Gamma(\mu-0.5)\Gamma(\sum\limits_{l=1}^L \mu_l+0.5 )}{\Gamma(\mu) \Gamma(\sum\limits_{l=1}^L \mu_l)} \times
\end{equation*}
\begin{equation}
\textstyle
F_D^{(2L-1)}(\sum\limits_{l=1}^L \mu_l+0.5, \mu_2-m_2,.,\mu_L-m_L,m_1,., m_L,\sum\limits_{l=1}^L \mu_l, 1-\frac{\theta_{11}}{\theta_{12}},.,1-\frac{\theta_{11}}{\theta_{1L}},1-\frac{\theta_{11}}{\theta_{21}},.,1-\frac{\theta_{11}}{\theta_{2L}}),
\label{eqn:evmkappamushadowedint}
\end{equation}

\fi

where $\mu>0.5$, $\theta_{1l}=\frac{1}{\mu_l(1+\kappa_l)}$, $\theta_{2l}=\frac{\mu_l \kappa_l +m_l}{\mu_l(1+\kappa_l)m_l}$, $\theta_{1}=\frac{1}{\mu(1+\kappa)}$, $\theta_{2}=\frac{\mu \kappa +m}{\mu(1+\kappa)m}$. Here $F_D^{(2L-1)}(.)$ is the Lauricella's function of the fourth kind \cite{exton1976multiple} and ${}_2F_1(.)$ is the Gauss hypergeometric function \cite{srivastava1985multiple}.  
\end{theorem}

\begin{proof}
Refer Appendix \ref{app:theorem1}
\end{proof}

Note that the EVM expression given in \eqref{eqn:evmkappamushadowedint} is the most general expression for an interference limited scenario as both the desired signal and interferers experience $\kappa$-$\mu$ shadowed fading. The expression is given in terms of Gauss hypergeometric function ${}_2F_1(.)$ and Lauricella's function of the fourth kind $F_D^{(L)}(.)$, which can be easily computed. 

Now, we simplify the EVM expression given in \eqref{eqn:evmkappamushadowedint} for various special cases. We start with the case when interfering signals experience i.i.d.  $\kappa$-$\mu$ shadowed fading of parameters $(\mu_I, \kappa_I, m_I)$.

\begin{cor}
\label{cor1}
The EVM  when both the desired signal and interferers experience i. i. d. $\kappa$-$\mu$ shadowed fading  is given by
\ifCLASSOPTIONtwocolumn
\begin{equation*}
\frac{\Gamma(\mu-0.5)}{\Gamma(\mu)\sqrt{\frac{\theta_1}{\theta_{1I}}}} {}_2F_1 \left(\mu -0.5,m,\mu,\frac{\mu \kappa}{m+ \mu \kappa} \right) \frac{\Gamma(L \mu _I+0.5 )}{\Gamma(L \mu_I)} \times
\end{equation*}
\begin{equation}
\frac{{}_2F_1 (L \mu _I + 0.5, L m_I, L \mu _I, \frac{\mu_ I  \kappa _I}{\mu _I \kappa _I +m _I} )}{ (   \frac{\theta_{2}}{\theta_{1}} )^{m} }  \left(   \frac{\theta_{1I}}{\theta_{2I}} \right)^{m_I L}     \label{eqn:evmkappamushadowedintiid}
\end{equation}
\else
\begin{equation}
\frac{\Gamma(\mu-0.5)}{\Gamma(\mu)\sqrt{\frac{\theta_1}{\theta_{1I}}}} {}_2F_1 (\mu -0.5,m,\mu,\frac{\mu \kappa}{m+ \mu \kappa} )   \frac{ \Gamma(L \mu _I+0.5 ) {}_2F_1 (L \mu _I + 0.5, L m_I, L \mu _I, \frac{\mu_ I  \kappa _I}{\mu _I \kappa _I +m _I} )}{ \Gamma(L \mu_I) (   \frac{\theta_{2}}{\theta_{1}} )^{m} (\frac{\theta_{2I}}{\theta_{1I}})^{m_I L} }      \label{eqn:evmkappamushadowedintiid}
\end{equation}
 \fi
where  $\mu > 0.5$, $\theta_{1I}=\frac{1}{\mu_I(1+\kappa_I)}$, $\theta_{2I}=\frac{\mu_I \kappa_I +m_I }{\mu_I(1+\kappa_I)m_I}$, $\theta_{1}=\frac{1}{\mu(1+\kappa)}$, $\theta_{2}=\frac{\mu \kappa +m}{\mu(1+\kappa)m}$. 
\end{cor} 
\begin{proof}
Refer Appendix \ref{app:cor1}.
\end{proof}

The above EVM expression  is in terms of only the Gauss hypergeometric function which can be easily computed. Next we derive the EVM of an interference limited system when both desired signal and interferers experience $\kappa$-$\mu$ fading.
 
 \begin{cor}
When the desired signal experiences $\kappa$-$\mu$ fading of parameters ($\kappa$, $\mu$) and is independent of interferers that experience i. i. d $\kappa$-$\mu$ fading of parameters ($\kappa_I$, $\mu_I$), EVM is
\ifCLASSOPTIONtwocolumn
\begin{equation*}
\frac{  \Gamma(\mu-0.5)  \Gamma(L \mu_I +0.5 ) }{\Gamma(L \mu_I) \Gamma(\mu)} {}_1F_1(0.5, \mu, -\kappa \mu)   \times
\end{equation*}
\begin{equation}
{}_1F_1(-0.5, L \mu_I, -L \kappa_I \mu_I) \sqrt{\frac{\mu(1+\kappa)} {\mu_I(1+\kappa_I)}} ,
\label{eqn:evmkappamuintiid}
\end{equation}
\else
\begin{equation}
\frac{  \Gamma(\mu-0.5)  \Gamma(L \mu_I +0.5 ) }{\Gamma(L \mu_I) \Gamma(\mu)} {}_1F_1(0.5, \mu, -\kappa \mu) {}_1F_1(-0.5, L \mu_I, -L \kappa_I \mu_I)  \sqrt{\frac{\mu(1+\kappa)} {\mu_I(1+\kappa_I)}},
\label{eqn:evmkappamuintiid}
\end{equation}
\fi

where  $\mu > 0.5$, ${}_1F_1$ is the Kummer confluent hypergeometric function \cite{srivastava1985multiple}.
\label{cor:kappamu} 
 \end{cor}
 
 \begin{proof}
 Refer Appendix \ref{app:kappamu}.
\end{proof}
 
 The above EVM expression is in terms of only the Kummer confluent hypergeometric function which can also be easily computed. Next we derive the EVM of an interference limited system when both the desired signal and interferers experience Rician fading.
 
 \begin{cor}
When the desired signal experiences Rician fading of parameter $K$ and is independent of interferers that experience i. i. d Rician fading of parameters $K_I$, 
\ifCLASSOPTIONtwocolumn
\begin{equation*}
 EVM=\frac{  \sqrt{\pi}  \Gamma(L  +0.5 ) }{\Gamma(L ) }  {}_1F_1(0.5, 1, -K) \times
\end{equation*}
\begin{equation}
{}_1F_1(-0.5, L , -L K_I) \sqrt{\frac{1+K} {1+K_I}} .
\label{eqn:evmricianint}
\end{equation}
\else
 \begin{equation}
 EVM=\frac{  \sqrt{\pi}  \Gamma(L  +0.5 ) }{\Gamma(L ) }  {}_1F_1(0.5, 1, -K) {}_1F_1(-0.5, L , -L K_I) \sqrt{\frac{1+K} {1+K_I}} .
\label{eqn:evmricianint}
\end{equation}
\fi

 \end{cor}
 
 \begin{proof}
Rician fading is a special case of $\kappa$-$\mu$ fading. From Table I,  $\mu=1$, $\kappa=K$, $\mu_I=1$, $\kappa_I=K_I$ . Substituting these in (\ref{eqn:evmkappamuintiid}), EVM is obtained.
 \end{proof}
 
 \ifCLASSOPTIONtwocolumn
 \begin{figure}
\centering
\includegraphics[height=3.5in,width=3.5in]{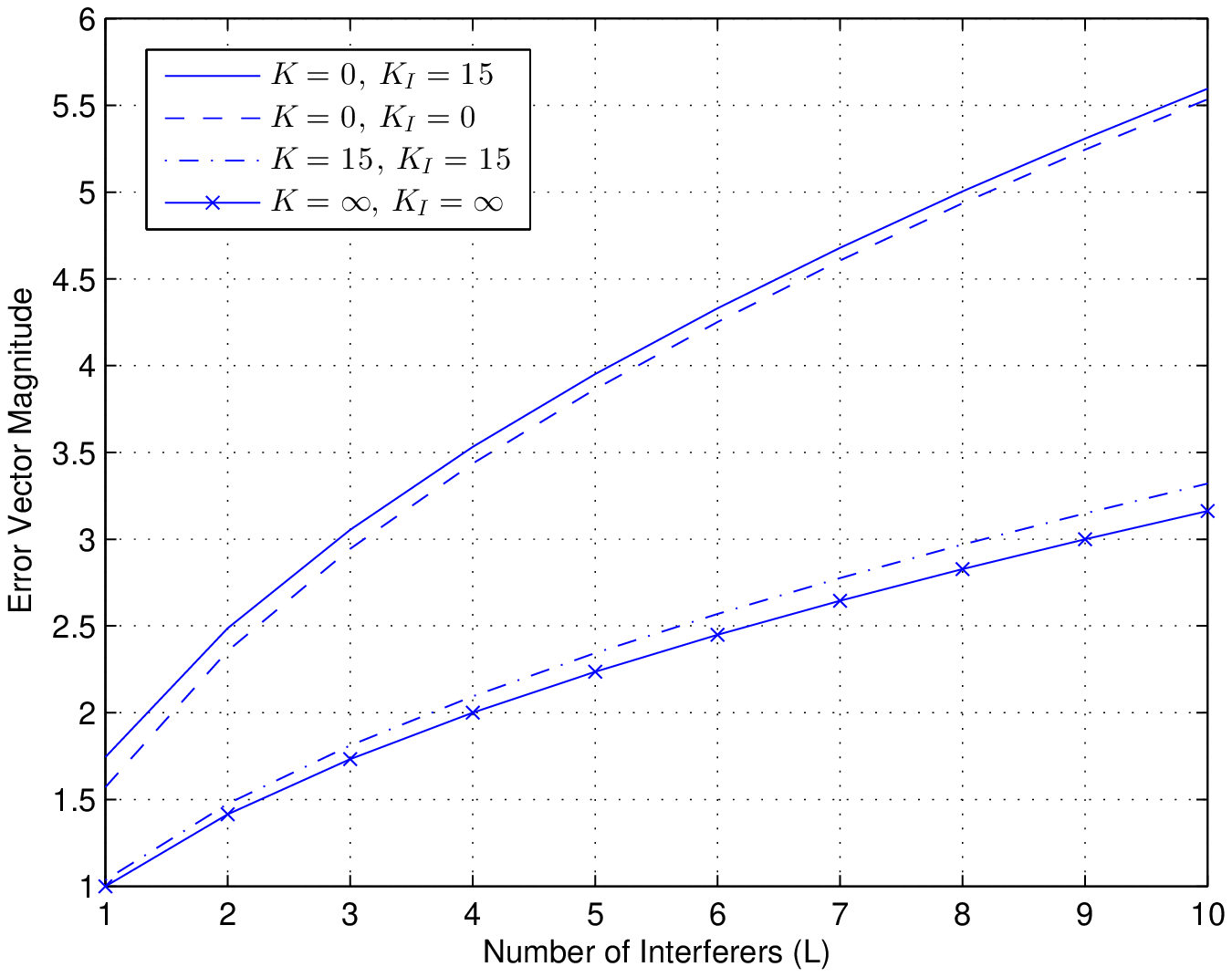}
\caption{EVM of an interference limited system when the interferers and desired signal experience Rician fading}
\label{fig:nov26}
\end{figure}

\else
 \begin{figure}
\centering
\includegraphics[height=4in,width=4in]{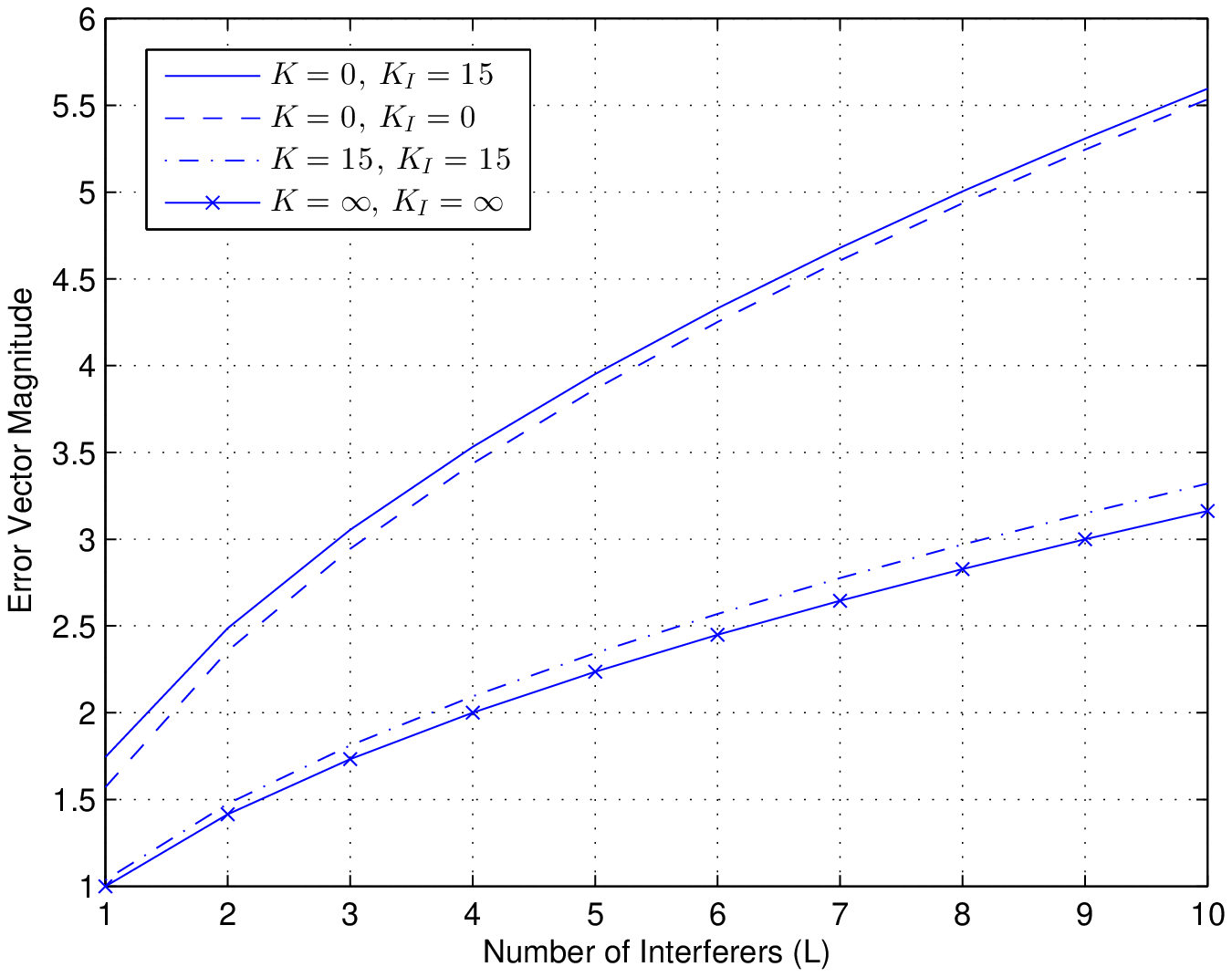}
\caption{EVM of an interference limited system when the interferers and desired signal experience Rician fading}
\label{fig:nov26}
\end{figure}

\fi

 We will now study the impact of various parameters such as $L$, $K$, $K_I$ on the EVM given in (\ref{eqn:evmricianint}). Firstly, we analyze the impact of $L$ on EVM, when the interferers and desired signal do not experience fading, i.e., $K \rightarrow \infty$. As $K \rightarrow \infty,$
 \begin{equation}
 {}_1F_1(0.5, 1, -K) \sqrt{1+K} \rightarrow \frac{1}{\sqrt{\pi}},
\label{eqn:Klimits}
 \end{equation}
 and as   $K_I \rightarrow \infty$,
 \begin{equation}
 {}_1F_1(-0.5, L , -L K_I) \sqrt{\frac{1} {1+K_I}}  \rightarrow \frac{\sqrt{L} \Gamma(L)}{\Gamma(L+0.5)}.
  \label{eqn:KIlimits}
 \end{equation}
Substituting (\ref{eqn:Klimits}) and (\ref{eqn:KIlimits}) in (\ref{eqn:evmricianint}), as $K \rightarrow \infty$ and $K_I \rightarrow \infty$ \ie, as interferers and desired signals tend to experience no fading, EVM$=\sqrt{L}$, i.e.,  EVM approaches square root of number of interferers. Now, we analyze the impact of $L$ on EVM, when the interferers and desired signal experience Rayleigh fading, i.e., $K=0 $ and $K_I=0$. 
 Using the fact that,
 \begin{equation*}
 {}_1F_1(a, b , 0)=1
 \end{equation*}
 one can show that when both $K=0 $ and $K_I=0$, i.e., when both the desired signal and interferers  experience  Rayleigh fading, then EVM given in \eqref{eqn:evmricianint} reduces to 
 \begin{equation}
 EVM= \frac{  \sqrt{\pi}   \Gamma(L  +0.5 ) }{\Gamma(L ) }\label{evm_ray}
 \end{equation}
Now, using the following identity
 \begin{equation}
\frac{\Gamma(n+a)}{\Gamma(n+b)}=\frac{(1+\frac{(a-b) (a+b-1)}{2 n} +O(\frac{1}{n^2}))}{n^{b-a}},  \text{for large n},
\label{eqn:identity}
\end{equation}
 the EVM given in \eqref{evm_ray} can be rewritten as 
\begin{equation*}
EVM=\sqrt{\pi}\sqrt{L} \left(1-\frac{1}{8L}+O(L{^{-2}}) \right)  \text{ for large L},
\end{equation*}
So when both  interferers and desired signal experience Rayleigh fading, EVM is approximately proportional to the square root of number of interfering channels for large number of interfering channels. We have also plotted the  EVM with respect to the number of interferers for different combinations of $K$ and $K_I$ in Fig. \ref{fig:nov26} using the expression given in \eqref{eqn:evmricianint}. It can be observed that as the Rician parameter $K$ increases, i.e., as the strength of line of sight component increases for the desired signal, EVM decreases and it is expected. Also as $K_I$ increases, \ie, as interferers become stronger, EVM increases as expected. When $K=K_I=15$, EVM approaches square root of number of interferers \ie, it approaches the EVM attained when $K=K_I \rightarrow \infty$. 

 \begin{cor}
When the desired signal experiences Nakagami fading of parameter $m$ and is independent of interferers that experience i. i. d Nakagami fading of parameters $m_I$,
 \begin{equation}
EVM= \frac{  \Gamma(m-0.5)  \Gamma(L m_I  +0.5 ) }{\Gamma(m) \Gamma(L  m_I )  } \sqrt{ \frac{m  } {m_I }},
\label{eqn:evmnakagamiintiid}
\end{equation}
where $m > 0.5$.
 \end{cor}

 \begin{proof}
Note that Nakagami-m fading is a special case of $\kappa$-$\mu$ fading.  Therefore, substituting  $\mu=m$, $\kappa \rightarrow 0$, $\mu_I=m_I$, $\kappa_I \rightarrow 0$  in (\ref{eqn:evmkappamuintiid}), and using the identity ${}_1F_1(.,.,0)=1$,  EVM expression is obtained.
 \end{proof}
 
 \ifCLASSOPTIONtwocolumn
 \begin{figure}
\includegraphics[height=3.5in, width=4in]{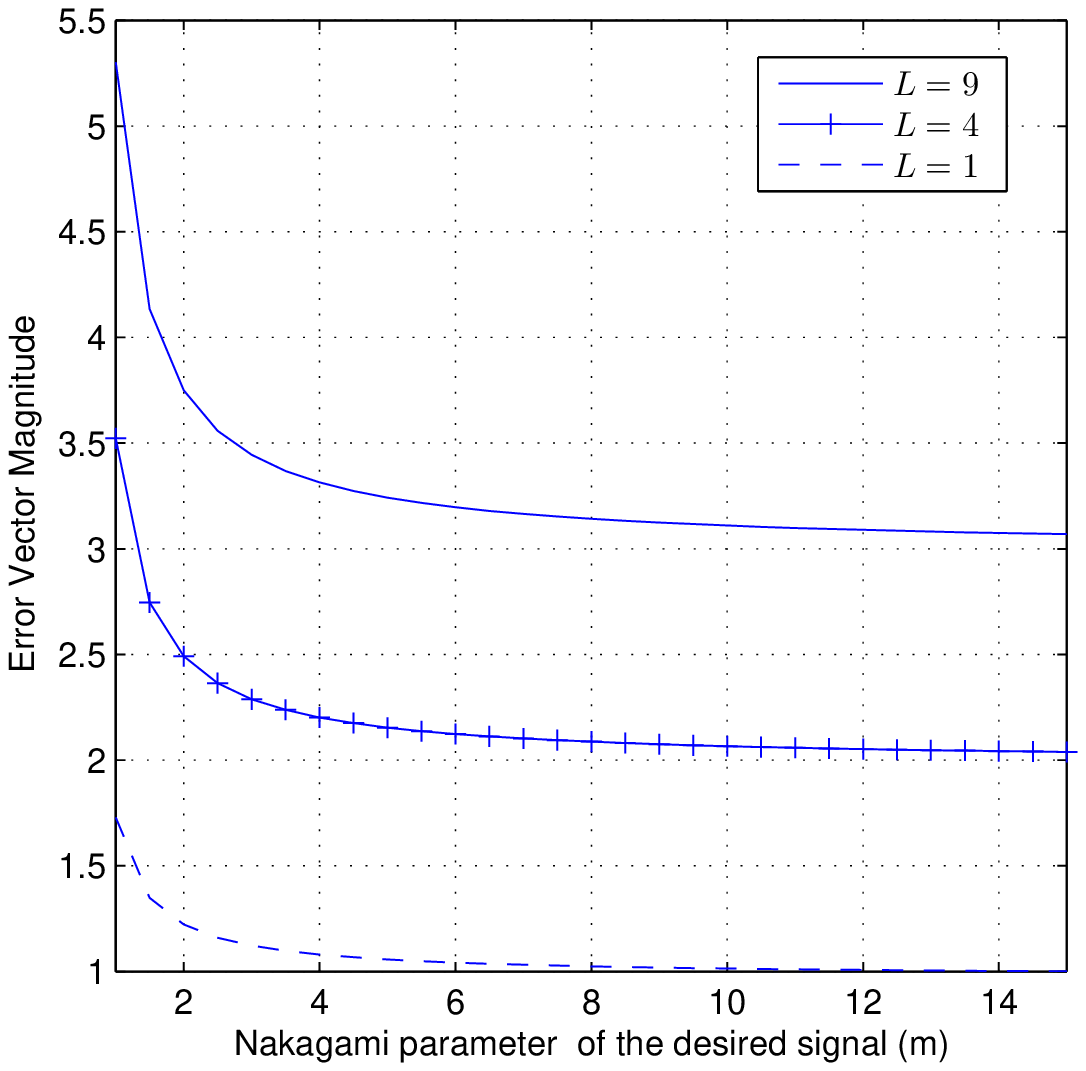}
\caption{EVM of an interference limited system when the interferers and desired signal experience Nakgamai fading. Here $m_I$=5 and $L$ denotes the number of interferers}
\label{fig:nakagamiparameter}
\end{figure}
 \else
\begin{figure}
\centering
\includegraphics[scale=1]{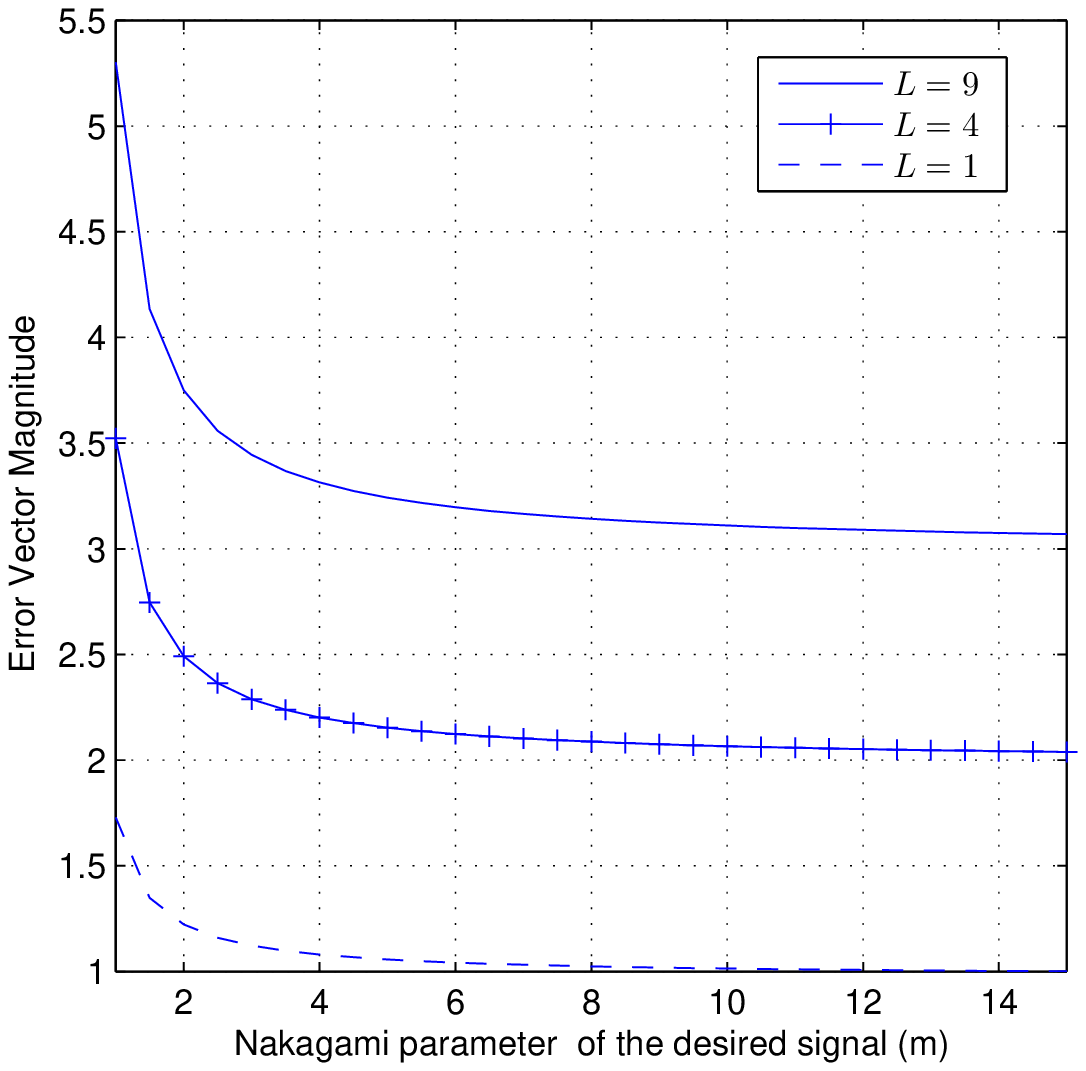}
\caption{EVM of an interference limited system when the interferers and desired signal experience Nakgamai fading. Here $m_I$=5 and $L$ denotes the number of interferers}
\label{fig:nakagamiparameter}
\end{figure}
\fi

 We will now study the impact of various parameters such as $L$, $m$, $m_I$ on the EVM given in \eqref{eqn:evmnakagamiintiid}. Firstly, we analyse the impact of $L$ on EVM, when the interferers and desired signal experience Nakagami-m fading with large value of $m$. Using the identity given in \eqref{eqn:identity}, one obtains,
 \begin{equation}
\frac{\sqrt{m} \Gamma(m-0.5)}{\Gamma(m)} = 1+\frac{3}{8m} +O(m^{-2}) \text{ for large $m$}
\label{eqn:largem}
 \end{equation}
 and for large $m_I$
 \begin{equation}
  \frac{ \Gamma(L m_I+0.5)}{\Gamma(L m_I) \sqrt{m_I}} = \sqrt{L} \left(1-\frac{1}{8L m_I} + O((L m_I)^{-2}) \right).
  \label{eqn:largemI}
 \end{equation}
Substituting \eqref{eqn:largem}, \eqref{eqn:largemI} in \eqref{eqn:evmnakagamiintiid} we observe that as $m \rightarrow \infty$ and $m_I \rightarrow \infty$, i.e., when signal and interferers experience no fading,  EVM approaches square root of number of interferers. This is similar to the one derived for Rician fading with parameter $K\rightarrow \infty$. When both the desired signal and interferers experience  Rayleigh fading, i.e., $m=m_I=1$, then EVM given in \eqref{eqn:evmnakagamiintiid} reduces to 
 \begin{equation}
 EVM= \frac{  \sqrt{\pi}   \Gamma(L  +0.5 ) }{\Gamma(L ) }
 \label{eqn:evm_ray2}
 \end{equation}
This is same as the one derived in (\ref{evm_ray}).
We have also plotted the  EVM with respect to the Nakagami parameter of the desired channel $m$ for different number of interferers as shown in Fig. \ref{fig:nakagamiparameter} using the expression given in \eqref{eqn:evmnakagamiintiid}. It can be observed that as the Nakagami parameter $m$ increases, i.e., as the number of clusters increase for the desired signal, EVM decreases as expected.  It can be observed that for higher value of $m$ and $m_I$, EVM converges to the square root of number of interfering channels as described before.

\subsection{Interference and Noise}
In this subsection, we derive the EVM in the presence of interference and AWGN of zero mean, variance $\sigma^2$, when the desired signal experiences $\kappa$-$\mu$ shadowed fading and interferers experience i.i.d. Nakagami-m fading. Note that  the EVM is given by 
\begin{equation*}
EVM=\int\limits_{0}^{\infty} \int\limits_{0}^{\infty} \sqrt{\frac{g_I+\sigma^2}{g_d}} f_{g_I}(g_I) f_{g_d}(g_d) \d g_I \d g_d.
\end{equation*}
Here, $g_I$ is the sum of interfering fading power RVs and hence it is Gamma distributed as it is assumed that the interferers experience  i.i.d. Nakagami-m fading of unit mean power, i.e., $g_I\sim G(Lm_I,\frac{1}{m_I})$ . The desired signal is affected by $\kappa$-$\mu$ shadowed fading power and the parameters of signal power $g_d$ are $(\kappa, \mu, m)$.  The following theorem provides the general expression of EVM in the presence of interference and noise.
\begin{theorem}
\label{thm2}
When the interferers experience i. i. d. Nakagami fading and the desired signal experiences $\kappa$-$\mu$ shadowed fading , EVM is 
\ifCLASSOPTIONtwocolumn
\begin{equation*}
(\sigma^2)^{m_I L+.5}  U \left(m_I L, \frac{3}{2} + m_I L , \sigma^2 m_I  \right) \sqrt{\mu(1+\kappa)} m_I^{m_I L} \times
\end{equation*}
\begin{equation}
 \left(\frac{m}{\mu\kappa+m} \right)^m \frac{\Gamma(\mu-0.5) }{\Gamma(\mu)} {}_2F_1 (\mu -\frac{1}{2},m,\mu,\frac{\mu \kappa}{m+ \mu \kappa} )   \label{eqn:evm_noise}
\end{equation}
\else
\begin{equation}
\frac{ U \left(m_I L, \frac{3}{2} + m_I L , \sigma^2 m_I  \right) \sqrt{\sigma^2 \mu(1+\kappa)} \left(\frac{m}{\mu\kappa+m} \right)^m \Gamma(\mu-0.5) {}_2F_1 (\mu -\frac{1}{2},m,\mu,\frac{\mu \kappa}{m+ \mu \kappa} )  }{(\sigma^2 m_I)^{-m_I L}  \Gamma(\mu) },
\label{eqn:evm_noise}
\end{equation}
\fi
where $\mu > 0.5$, ${}_2 F_1(.)$ is the Gauss hypergeometric function  and $U(.)$ is the Tricomi confluent hypergeometric function.
\label{thm:noise}
\end{theorem}

\begin{proof}
Refer Appendix \ref{app:theorem2}.
\end{proof}

 The above EVM expression is in terms of Tricomi confluent hypergeometric function and Gauss hypergeometric function, and both of them can be easily computed. Next we derive the EVM expression  when the desired signal experiences $\kappa$-$\mu$ fading.
\begin{cor}
When the desired signal experiences $\kappa$-$\mu$  fading of parameters $(\kappa, \mu)$, the EVM is given by
\ifCLASSOPTIONtwocolumn
\begin{equation*}
\left(\sigma^2\right)^{m_I L+.5}  U \left(m_I L, \frac{3}{2} + m_I L , \sigma^2 m_I  \right) \sqrt{\mu(1+\kappa) }  \times
\end{equation*}
\begin{equation}
\frac{\Gamma(\mu-0.5) }{\Gamma(\mu)}{}_1F_1 \left(\frac{1}{2},\mu,-\kappa \mu \right)  m_I^{m_I L},
\label{eqn:evmkappamu}
\end{equation}
\else
\begin{equation}
\left(\sigma^2\right)^{m_I L+.5}  U \left(m_I L, \frac{3}{2} + m_I L , \sigma^2 m_I  \right) \sqrt{\mu(1+\kappa) }  \frac{\Gamma(\mu-0.5) }{\Gamma(\mu)}{}_1F_1 \left(\frac{1}{2},\mu,-\kappa \mu \right)  m_I^{m_I L},
\label{eqn:evmkappamu}
\end{equation}
\fi
where $\mu > 0.5$.
\end{cor}
\begin{proof}
$\kappa$-$\mu$  fading is a special case of $\kappa$-$\mu$ shadowed fading when $m \rightarrow \infty$ \cite{6594884}. Combining (\ref{eqn:limit1}) and (\ref{eqn:evm_noise}) and using the transformation $e^{-z} {}_1F_1(b-a,b,z)={}_1F_1(a,b,-z)$, EVM is obtained.
\end{proof}
Again, the above EVM expression is in terms of Tricomi confluent hypergeometric function and Kummer confluent hypergeometric function which can be easily computed. 
Rician fading is a special case of $\kappa$-$\mu$ fading for $\mu$=1, $\kappa$=$K$. Hence the EVM when desired signal experience Rician fading is given by
\ifCLASSOPTIONtwocolumn
\begin{equation}
  \frac{  U \left(m_I L, \frac{3}{2} + m_I L , \sigma^2 m_I  \right)  {}_1F_1 \left(\frac{1}{2},1,-K \right) }{ (m_I \sigma^2)^{-L m_I} (\sigma^2 \pi(1+K))^{-0.5}  },
\label{eqn:evmriciannoise}
\end{equation}
\else
\begin{equation}
\left(\sigma^2\right)^{m_I L+.5}  U \left(m_I L, \frac{3}{2} + m_I L , \sigma^2 m_I  \right) \sqrt{\pi(1+K) } {}_1F_1 \left(\frac{1}{2},1,-K \right)  m_I^{m_I L},
\label{eqn:evmriciannoise}
\end{equation}
\fi

Now, we derive the EVM expression  when the desired signal experiences Nakagami-m fading.
\begin{cor}
When the desired signal experiences Nakagami-$m$ fading, the EVM is given by

\begin{equation}
\frac{\left(\sigma^2 m_I \right)^{m_I L}  U \left(m_I L, \frac{3}{2} + m_I L , \sigma^2 m_I  \right) \Gamma(m-0.5)}{  (m \sigma^2)^{-0.5} \Gamma(m)},
\label{eqn:evmnakagami}
\end{equation}

where $m > 0.5$.
\end{cor}
\begin{proof}
Nakagami is a special case of $\kappa$-$\mu$ fading for parameters $\kappa \rightarrow 0$, $\mu=m$ \cite{6594884}. Substituting these in (\ref{eqn:evmkappamu}) and using the fact that ${}_1F_1(a,b,0)=1$, EVM expression is obtained. 
\end{proof}
Note that the EVM expression is very simple when the desired signal experience Nakagami-m fading. Similarly by putting $m=1$ and $m_I=1$ in \eqref{eqn:evmnakagami}, one can obtain the EVM expression when both desired signal and interferers experience Rayleigh fading and it is given by
 \begin{equation}
EVM=
\left(\sigma^2\right)^{L+.5} U \left( L, \frac{3}{2} +  L , \sigma^2 \right) \sqrt{\pi}   .
\label{eqn:evmrayleigh}
\end{equation}

\section{Numerical Results}
In this Section, we verify the EVM expressions derived for  both the interference limited and AWGN+interference case for different fading distributions through Monte Carlo simulation. We use a BPSK modulation based system over different fading channels for the simulation. 
Fig. \ref{fig:kappamunoise} plots the EVM with respect to SNR when both interference and noise are present. Here  the interferers and desired signal experience  Nakagami-m fading and $\kappa$-$\mu$ shadowed fading, respectively. Note that the theoretical results are plotted using    \eqref{eqn:evm_noise}. Firstly, it can be observed that the simulation results match with the theoretical results. Secondly, it can be seen that as SNR increases, EVM decreases and it is true for all combinations of fading parameters. Interestingly, it can be also observed that as any of the fading parameter, i.e., $\kappa$, $\mu$ or $m$ increases, EVM decreases for a given SNR. The reason for such behavior is as follows: the parameter $\kappa$ denotes the strength of the line of sight component in the desired signal. Now, when the parameter $\kappa$ increases, the strength of the line of sight component increases and hence EVM decreases. Similarly, the parameter $\mu$ denotes the number of clusters in the desired channel and hence when $\mu$ increases, EVM decreases. Note that in the $\kappa$-$\mu$ shadowed fading, all the dominant components are subjected to the common shadowing fluctuation and that is represented by the power-normalized random amplitude $\xi$ and this $\xi$ is Nakagami-m random variable with shaping parameter $m$. Now, when the shadowing parameter $m$ increases, the shadowing fluctuation of all the dominant components decreases and hence EVM decreases.   

Fig. \ref{fig:fadingnoise} plots the EVM with respect to SNR when the desired signal experience different fading, such as     Rician, Nakagami  and Rayleigh. In all the cases the interferers experience Nakagami-m fading. The theoretical results for    Rician,   Nakagami  and Rayleigh fading are plotted using \eqref{eqn:evmriciannoise}, \eqref{eqn:evmnakagami} and \eqref{eqn:evmrayleigh}, respectively.  It can be seen that the simulation results match with the theoretical results and as expected with increase in SNR, the EVM decreases. 
 
 Fig. \ref{fig:fadinginterference} plots the EVM with respect to number of interferers  for an interference limited scenario.
The theoretical results for Rayleigh, Rician and $\kappa$-$\mu$   shadowed  fading are plotted using \eqref{eqn:evm_ray2}, \eqref{eqn:evmricianint}, \eqref{eqn:evmkappamushadowedintiid}, respectively.  It can be seen that the simulation results match with the theoretical results and as expected as number of interferers increases, EVM increases.

 Fig. \ref{fig:fadingnakagamilimiting}  plots the EVM with respect to SNR.  Here both the systems, i.e.,  interference + AWGN system and interference  limited system have been considered. The EVM for both interference limited scenario and interference+AWGN is plotted using \eqref{eqn:evmnakagamiintiid} and \eqref{eqn:evmnakagami}.  It can be observed that when SNR is significantly low, i.e., noise limited case, there is significant gap between the EVM of the interference+noise system and EVM of the interference limited system. The gap decreases as the SNR increases. Moreover, the EVM of the  interference+noise system approaches the EVM of the interference limited system at very high SNR as expected.

 \ifCLASSOPTIONtwocolumn
 \begin{figure}
\centering
\includegraphics[height=3.5in,width=4in]{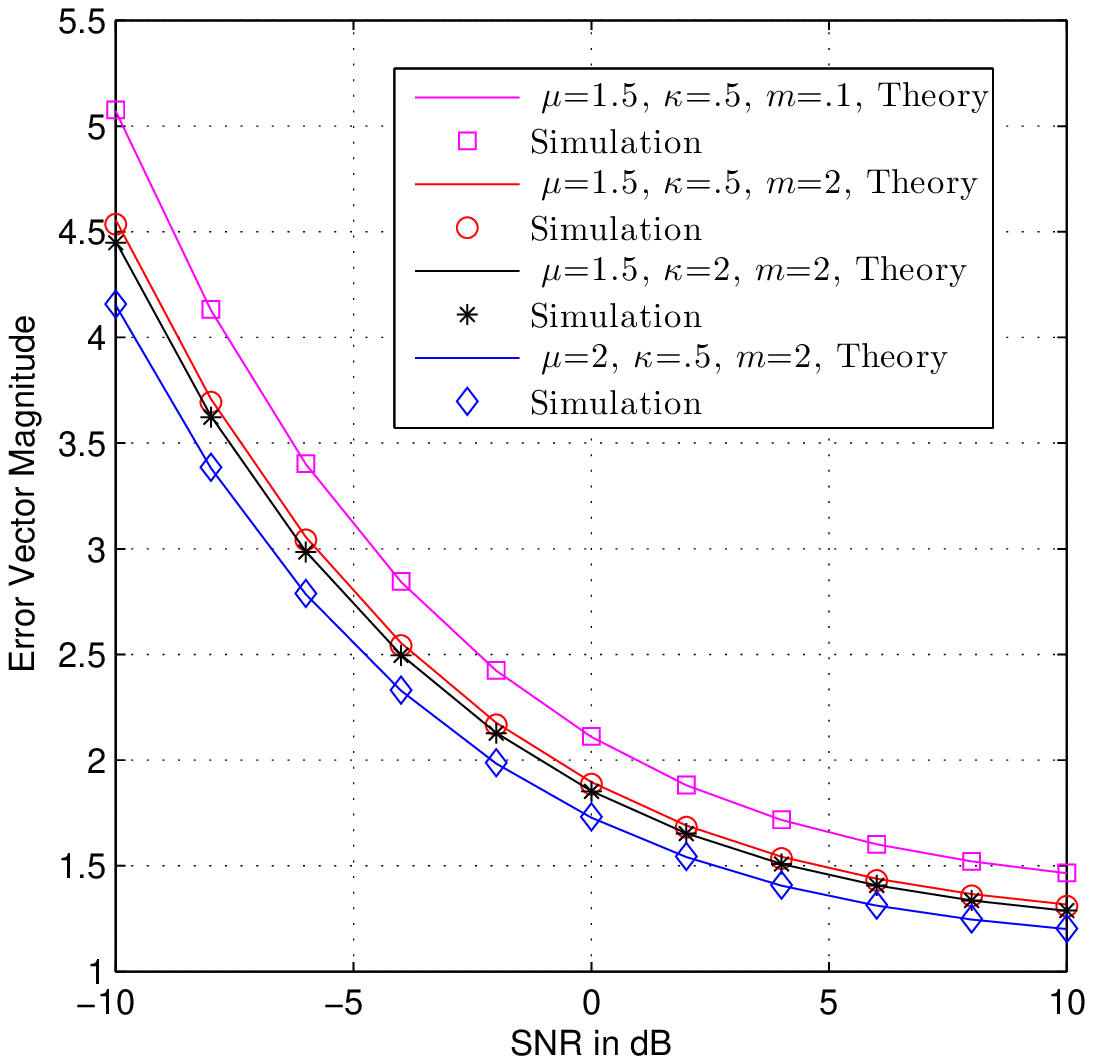}
\caption{Theoretical and simulated  EVM for different parameters of $\kappa$-$\mu$ shadowed fading when both interference and noise are present. Here $m_I$=1, $L$=1.}
\label{fig:kappamunoise}
\end{figure}

 \begin{figure}
\centering
\includegraphics[height=3.5in,width=4in]{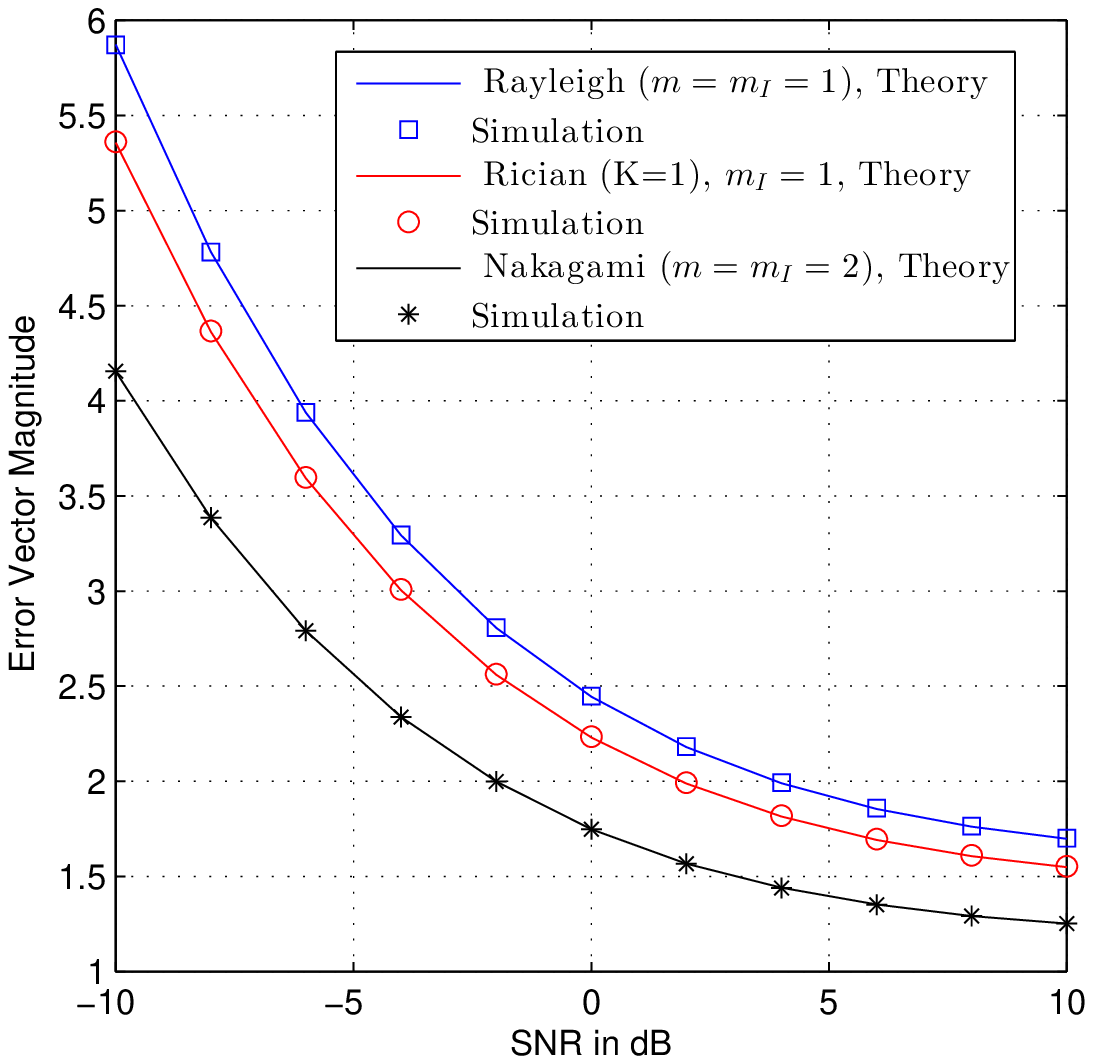}
\caption{Theoretical and simulated EVM for different fading distributions when both interference and noise are present. }
\label{fig:fadingnoise}
\end{figure}

 \begin{figure}
\centering
\includegraphics[height=3.5in,width=4in]{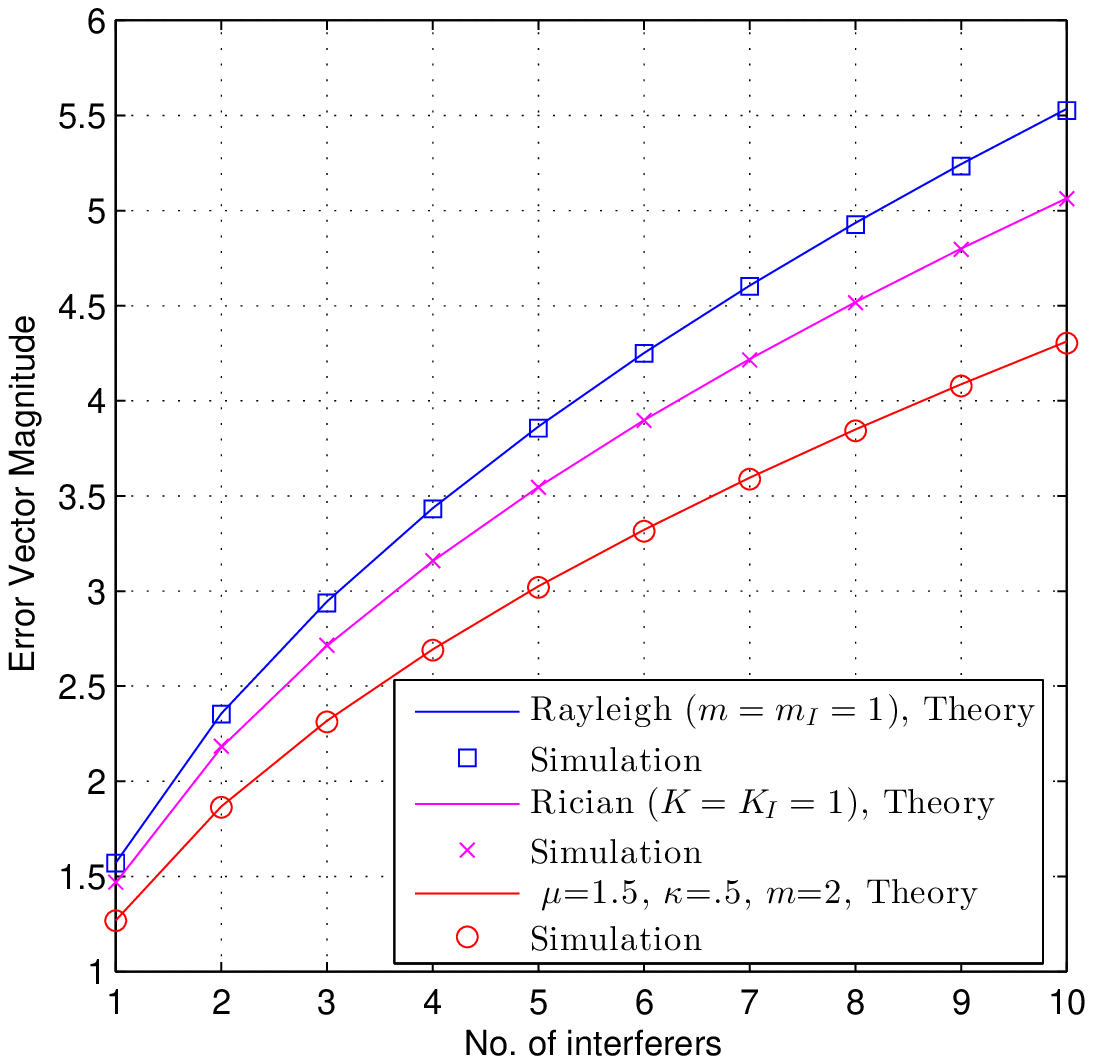}
\caption{Theoretical and simulated EVM for different fading distributions in an interference limited system  interferer signals experience i.i.d. fading.}
\label{fig:fadinginterference}
\end{figure}

\begin{figure}
\centering
\includegraphics[height=3.5in,width=4in]{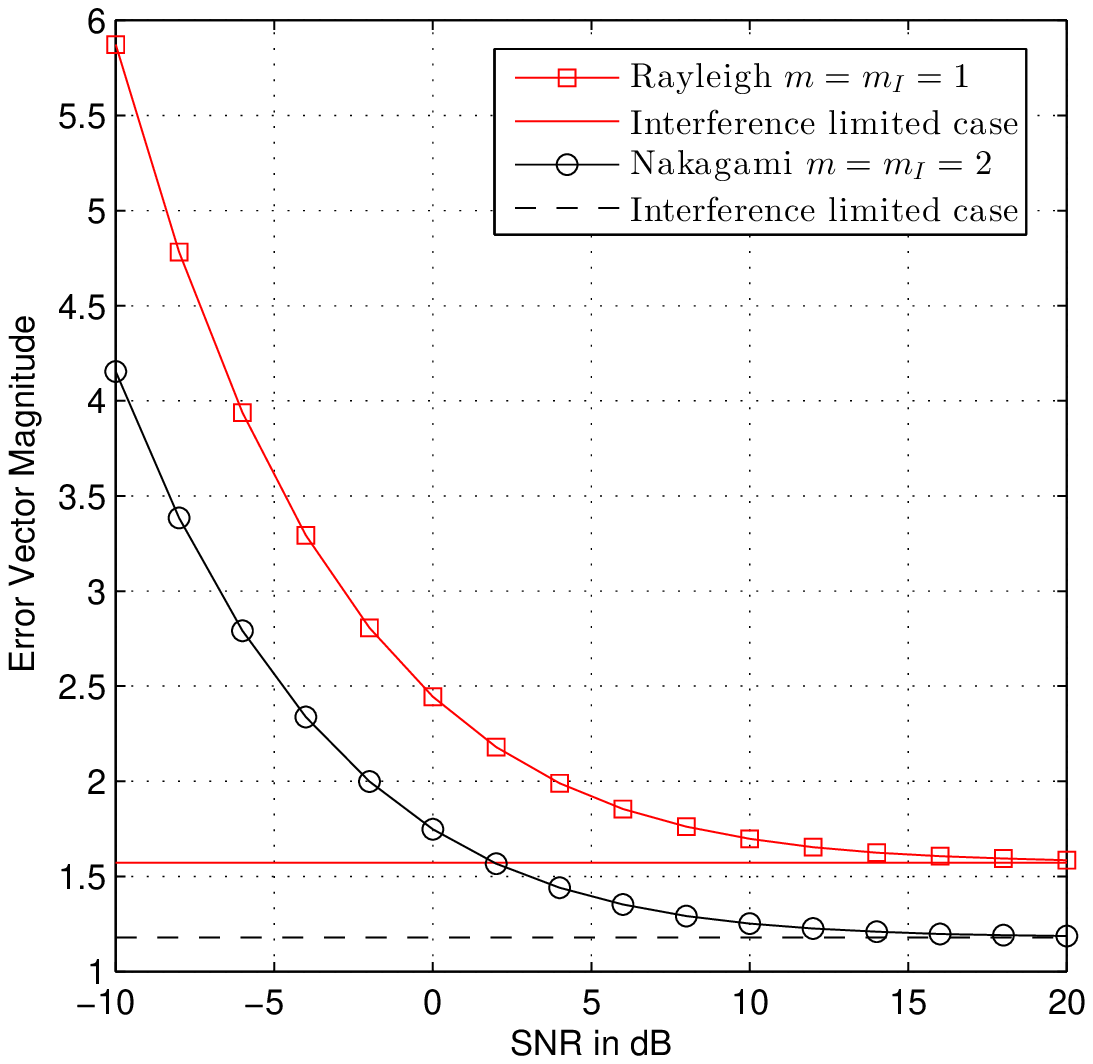}
\caption{EVM for Nakagami fading when both interference + AWGN and interference limited system are considered. Here $L$=1.}
\label{fig:fadingnakagamilimiting}
\end{figure}

\else

 \begin{figure}
\centering
\includegraphics[scale=1]{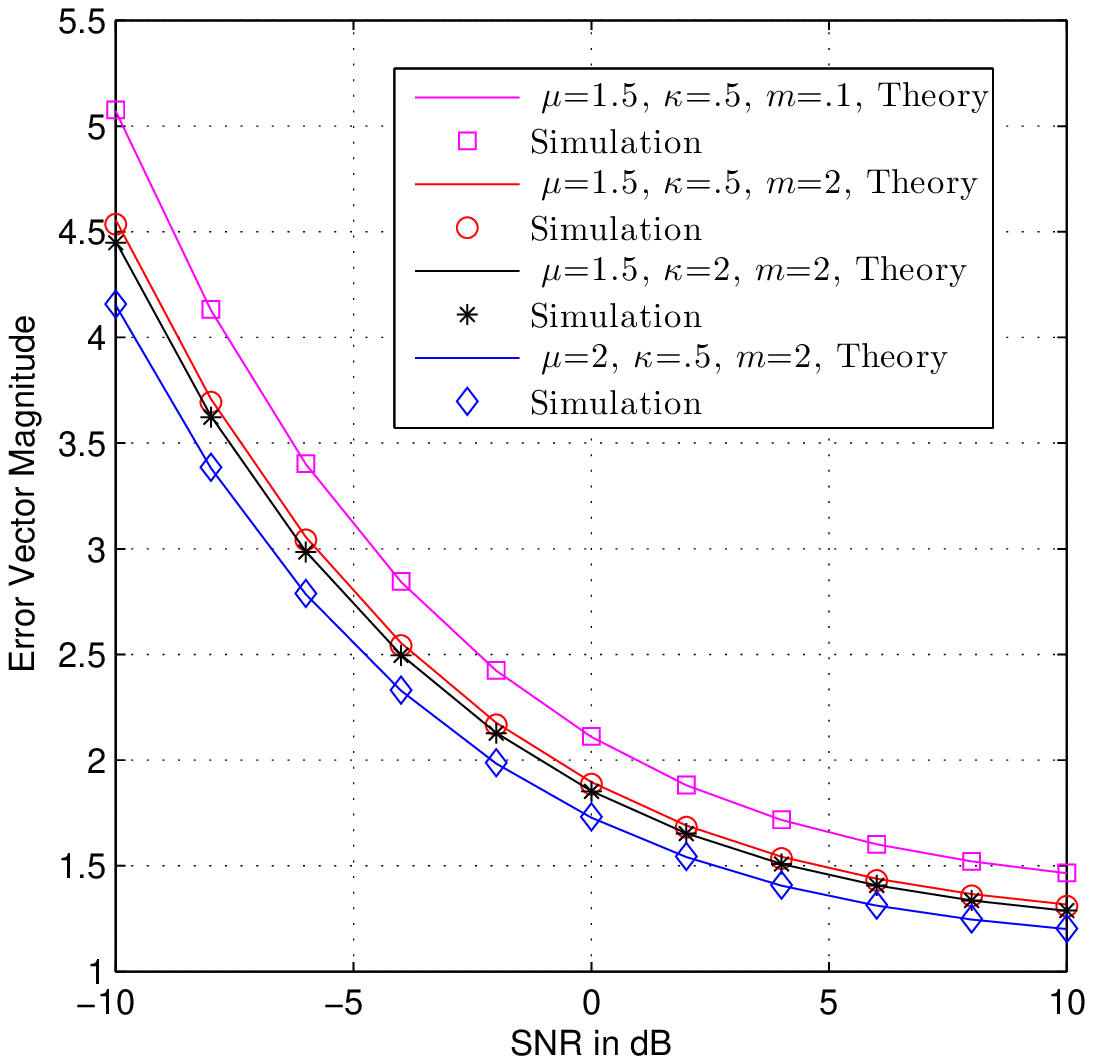}
\caption{Theoretical and simulated  EVM for different parameters of $\kappa$-$\mu$ shadowed fading when both interference and noise are present. Here $m_I$=1, $L$=1.}
\label{fig:kappamunoise}
\end{figure}

 \begin{figure}
\centering
\includegraphics[scale=1]{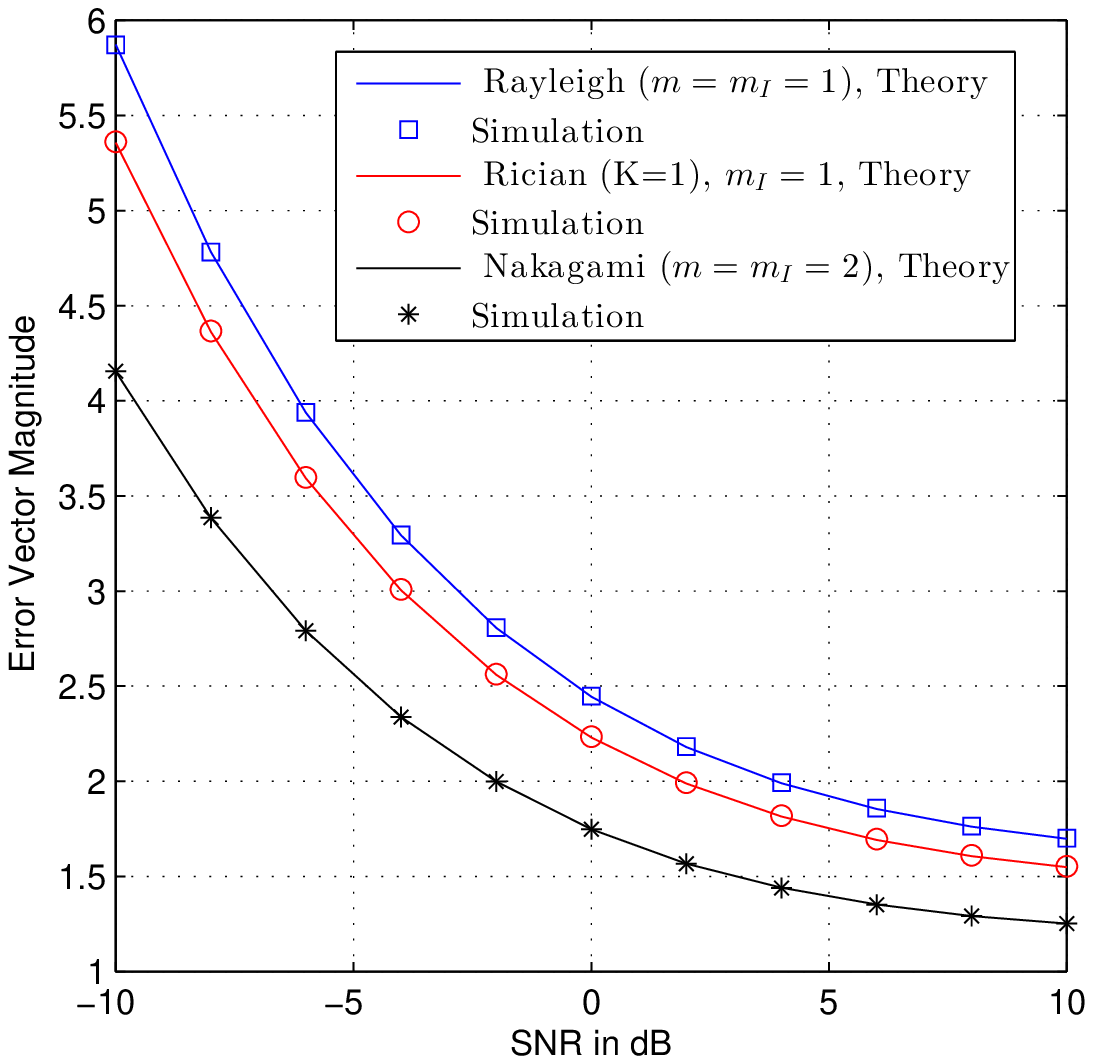}
\caption{Theoretical and simulated EVM for different fading distributions when both interference and noise are present. }
\label{fig:fadingnoise}
\end{figure}

 \begin{figure}
\centering
\includegraphics[scale=1]{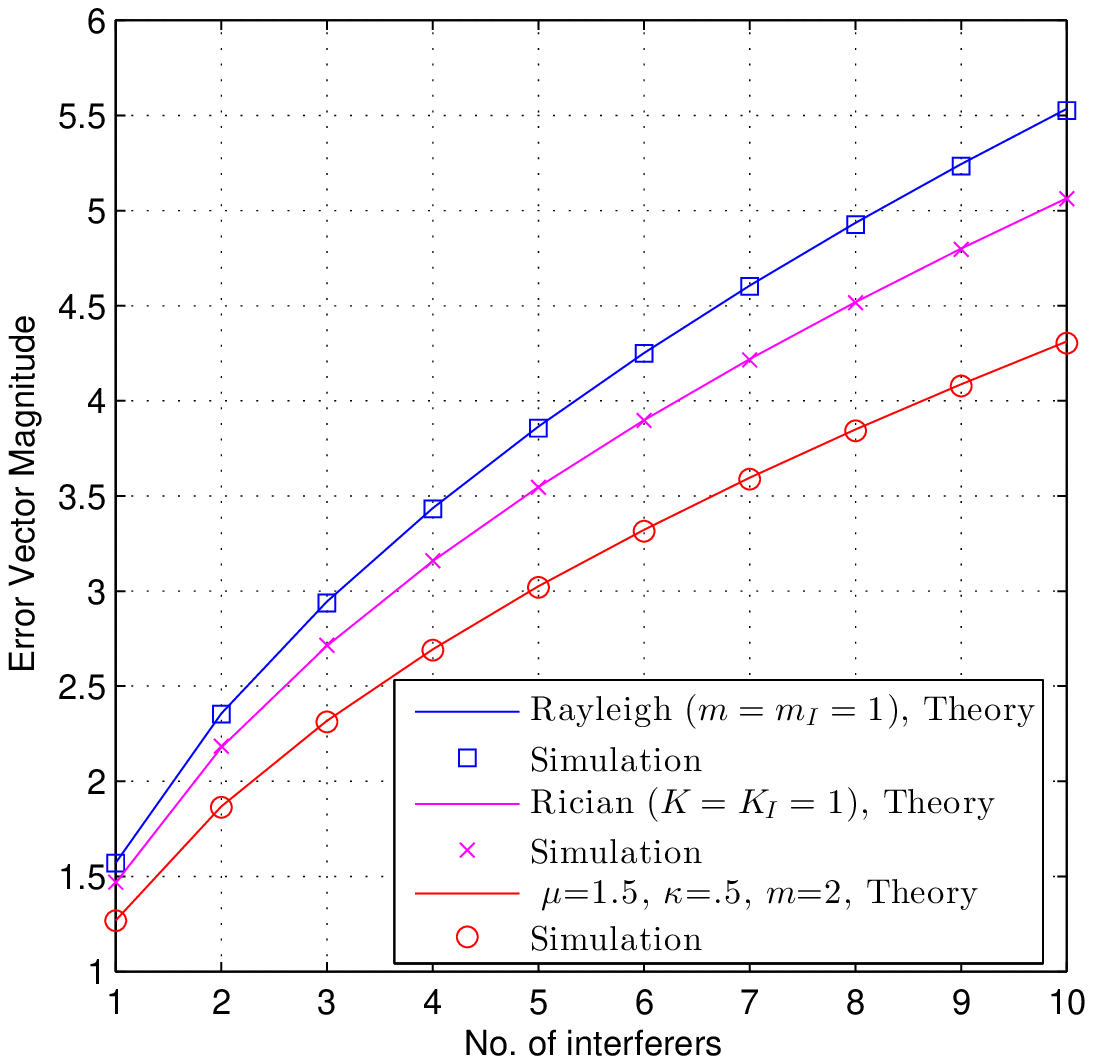}
\caption{Theoretical and simulated EVM for different fading distributions in an interference limited system  interferer signals experience i.i.d. fading.}
\label{fig:fadinginterference}
\end{figure}

\begin{figure}
\centering
\includegraphics[scale=1]{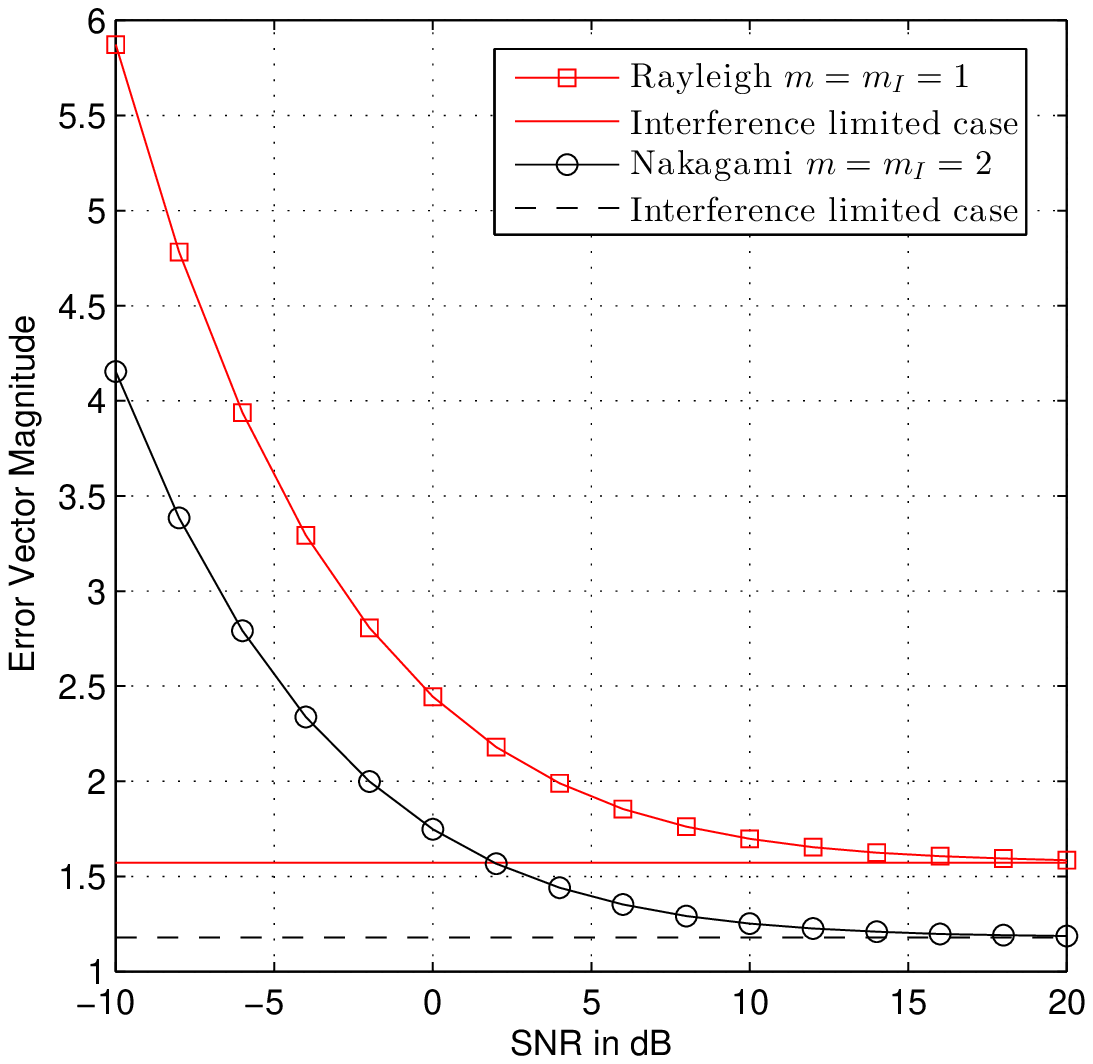}
\caption{EVM for Nakagami fading when both interference + AWGN and interference limited system are considered. Here $L$=1.}
\label{fig:fadingnakagamilimiting}
\end{figure}

\fi

\section{Conclusion}
In this paper, we have derived the closed form expression for EVM in an interference limited system when both the desired signal and interfering signals experience independent, non-identically distributed $\kappa$-$\mu$ shadowed fading. This has been simplified for many special cases such as Rayleigh, Nakagami, Rician, $\kappa$-$\mu$ fading etc. We have also shown that in an interference limited system, EVM is equal to the square root of number of interferers when the interferers and desired signal do not experience fading. EVM expression is also derived in the presence of AWGN when the desired signal experiences $\kappa$-$\mu$ shadowed fading and interferers experience independent, identically distributed Nakagami fading. Note that in this paper, we have studied the EVM only for SISO scenario. In future, it would be interesting to characterize the   EVM in multi-antenna systems and in the presence of correlated interferers.

\begin{appendices}
\section{}
\section*{Proof of Theorem \ref{thm1}}
\label{app:theorem1}
Here we derive the EVM of an interferer limited system when the desired signal experiences $\kappa$-$\mu$ shadowed fading of unit mean power with parameters ($\kappa$, $\mu$, $m$) and interferers experience independent non identical fading with the $l$th interferer experiencing $\kappa$-$\mu$ shadowed fading of unit mean power with parameters ($\kappa_l$, $\mu_l$, $m_l$).
Recall from (\ref{eqn:evm}) that EVM of an interference limited system is 
\begin{equation*}
\EVM=\int\limits_{0}^{\infty} \int\limits_{0}^{\infty} \sqrt{\frac{g_I}{g_d}} f_{g_I}(g_I)  f_{g_d}(g_d) \d g_I   \d g_d. 
\end{equation*}

As the double integral is in a separable form ,

\begin{equation}
\EVM=\int\limits_{0}^{\infty} \sqrt{g_I} f_{g_I}(g_I) \d g_I  \int\limits_{0}^{\infty} \sqrt{\frac{1}{g_d}} f_{g_d}(g_d) \d g_d. 
\label{eqn:evm_kappamu1}
\end{equation}

First we evaluate the second part of the above EVM expression \ie, $\E \left(\sqrt{\frac{1}{g_d}} \right)$. As the desired signal is $\kappa$-$\mu$ shadowed, from (\ref{eqn:gdpdf}) we know that 
\begin{equation*}
f_{g_d}(g_d)=\frac{\theta_{1}^{m-\mu}g_d^{\mu-1} } {\theta_{2}^{m}\Gamma(\mu)}e^{-\frac{g_d}{\theta_{1}}} {}_1 F_1 \left(m,\mu,\frac{(\theta_{2}-\theta_{1})g_d}{\theta_{1}\theta_{2}} \right),
\end{equation*}
where $\theta_1=\frac{1}{\mu(1+\kappa)}$, $\theta_2=\frac{\mu \kappa+m }{\mu(1+\kappa)m}$.
So  $\int\limits_{0}^{\infty} \sqrt{\frac{1}{g_d}} f_{g_d}(g_d) \d g_d$ is
\begin{equation}
\int\limits G_2 g_d^{\mu-1.5} e^{-\mu (1+\kappa) g_d} {}_1F_1 (m,\mu, \frac{\mu^2 \kappa (1+\kappa) g_d}{\mu \kappa+m}) \d g_d, 
 \label{eqn:Egdinverse}
\end{equation} 
 where $G_2=\frac{\mu^{\mu} m^m (1+\kappa)^{\mu}}{\Gamma(\mu)  (\mu \kappa+m)^m}.$

Substitute $t$=$\mu (1+\kappa) g_d$ in (\ref{eqn:Egdinverse}) and by using the identity in \cite[Page.17]{exton1976multiple} \ie, ${}_2F_1(a,b,c,x)=\frac{1}{\Gamma(a)}\int\limits_{0}^{\infty} e^{-t} t^{a-1} {}_1F_1(b;c;x t) \d t, Re(a)>0, $ $\int\limits_{0}^{\infty} \sqrt{\frac{1}{g_d}} f_{g_d}(g_d) \d g_d$ is

\begin{equation}
\frac{G_2   \Gamma(\mu-0.5) {}_2F_1 \left(\mu -\frac{1}{2},m,\mu,\frac{\mu \kappa}{m+ \mu \kappa} \right)}{\left(\mu(1+\kappa) \right)^{\mu-0.5}}, \text{for $\mu > $ 0.5}
 \label{eqn:g2_kappamu1}
 \end{equation}

Now we will evaluate the first part of the EVM expression \ie, $\E(\sqrt{g_I})$. Using the $\pdf$ of $g_I$ from (\ref{eqn:gIpdf}), $\int\limits_{0}^{\infty} \sqrt{g_I} f_{g_I}(g_I) \d g_I$ is
\ifCLASSOPTIONtwocolumn
\begin{equation*}
\int\limits_{0}^{\infty} G_1 g_I^{(\sum\limits_{l=1}^L \mu_l -0.5)} \Phi_2^{2 L}(\mu_1-m_1,.., \mu_L-m_L,m_1,..,m_L; \sum\limits_{l=1}^L \mu_l ; 
\end{equation*}
\begin{equation*}
-a_1  g_I,..,-a_L g_I,-b_1 g_I,..,-b_L g_I )   \d g_I,
\label{eqn:EgI}
\end{equation*}
\else
\begin{equation*}
\int\limits_{0}^{\infty} \frac{G_1 \Phi_2^{2 L}(\mu_1-m_1,.., \mu_L-m_L,m_1,..,m_L; \sum\limits_{l=1}^L \mu_l ; -a_1  g_I,..,-a_L g_I,-b_1 g_I,..,-b_L g_I ) }{  g_I^{-(\sum\limits_{l=1}^L \mu_l -0.5)} }\d g_I,
\label{eqn:EgI}
\end{equation*}
\fi
where $a_l=\mu_l(1+\kappa_l),$ $b_l=a_l \frac{m_l}{\mu_l \kappa_l +m_l},$ $G_1= \frac{ \prod\limits_{l=1}^L b_l^{m_l} a_l^{\mu_l-m_l}}{\Gamma(\sum\limits_{l=1}^L \mu_l)}$.

Using the transformation given in \cite[P.177]{exton1976multiple} that $e^{-x_i}\Phi_2^{(n)}(b_1, \cdots, b_n;c;x_1, \cdots,x_n)$ is equivalent to
$\Phi_2^{(n)}(b_1, \cdot,b_{i-1},c-b_1-\cdot-b_n,b_{i+1},\cdot, b_n;c;x_1-x_i,\cdots  x_{i-1}-x_i,-x_i,x_{i+1}-x_i,\cdot,x_n-x_i),$ $\int\limits_{0}^{\infty} \sqrt{g_I} f_{g_I}(g_I)  \d g_I$ is
\ifCLASSOPTIONtwocolumn
\begin{equation*}
\int\limits_{0}^{\infty} e^{-a_1 g_I} G_1 g_I^{(\sum\limits_{l=1}^L \mu_l -0.5)} \Phi_2^{2 L}(0,\mu_2-m_2,.., \mu_L-m_L,m_1,..,m_L; 
\end{equation*}
\begin{equation*}
\sum\limits_{l=1}^L \mu_l ; a_1  g_I, a_1 g_I-a_2 g_I,..,a_1 g_I-b_L g_I )   \d g_I .
\end{equation*}
\else
\begin{equation*}
\int\limits_{0}^{\infty} \frac{e^{-a_1 g_I} G_1 \Phi_2^{2 L}(0,\mu_2-m_2,.., \mu_L-m_L,m_1,..,m_L; \sum\limits_{l=1}^L \mu_l ; a_1  g_I, a_1 g_I-a_2 g_I,..,a_1 g_I-b_L g_I ) }{ g_I^{-(\sum\limits_{l=1}^L \mu_l -0.5)}} \d g_I .
\end{equation*}
\fi

Note that if one of the numerator parameters of the series expansion of $\Phi_2^{(L)}$  goes to zero, then $\Phi_2^{(L)}$
 becomes $\Phi_2^{(L-1)}$ and hence the above $\Phi_2^{(2L)}$ will become $\Phi_2^{(2L-1)}$ with appropriate parameters. So $\int\limits_{0}^{\infty} \sqrt{g_I} f_{g_I}(g_I)  \d g_I$ is
 \ifCLASSOPTIONtwocolumn
  \begin{equation*}
\int\limits_{0}^{\infty} e^{-a_1 g_I} g_I^{(\sum\limits_{l=1}^L \mu_l -0.5)}  G_1 \Phi_2^{2 L-1}(\mu_2-m_2,.., \mu_L-m_L,m_1,..,m_L; 
\end{equation*}
 \begin{equation}
\sum\limits_{l=1}^L \mu_l ; a_1 g_I-a_2 g_I,..,a_1 g_I-b_L g_I )  \d g_I.
\label{eqn:EgI2}
\end{equation}
 \else
 \begin{equation}
 \small
\int\limits_{0}^{\infty} \frac{e^{-a_1 g_I} G_1 \Phi_2^{2 L-1}(\mu_2-m_2,.., \mu_L-m_L,m_1,..,m_L; \sum\limits_{l=1}^L \mu_l ; a_1 g_I-a_2 g_I,..,a_1 g_I-b_L g_I ) }{ g_I^{-(\sum\limits_{l=1}^L \mu_l -0.5)}} \d g_I.
\label{eqn:EgI2}
\end{equation}
\fi
 
 In order to  simplify this, we use the following relationship between  confluent Lauricella function and Lauricella's function of the fourth kind \cite[P. 286, Eq. 43]{srivastava1985multiple} 
 \ifCLASSOPTIONtwocolumn
 \begin{equation*}
F_D^{(N)}[\alpha,\beta_1,\cdots,\beta_N;\gamma;x_1,\cdots, x_N]=
\label{eq:relationship}
\end{equation*}
\begin{equation*}
\frac{1}{\Gamma(\alpha)}\int\limits_{0}^{\infty}e^{-t}t^{\alpha-1}\Phi_2^{(N)}[\,\beta_1,\cdots,\beta_N,\gamma;x_1t,\cdots,x_Nt]\text{d}t.
\end{equation*}
 \else
\begin{equation*}
F_D^{(N)}[\alpha,\beta_1,\cdots,\beta_N;\gamma;x_1,\cdots, x_N]=\frac{1}{\Gamma(\alpha)}\int\limits_{0}^{\infty}e^{-t}t^{\alpha-1}\Phi_2^{(N)}[\,\beta_1,\cdots,\beta_N,\gamma;x_1t,\cdots,x_Nt]\text{d}t.
\label{eq:relationship}
\end{equation*}
\fi

By using the above identity and substituting $t$=$a_1 g_I$ in \eqref{eqn:EgI2}, $\int\limits_{0}^{\infty} \sqrt{g_I} f_{g_I}(g_I
) \d g_I$ is
\ifCLASSOPTIONtwocolumn
\begin{equation*}
\frac{G_1  \Gamma(\sum\limits_{l=1}^L \mu_l+0.5 )}{a_1^{(\sum\limits_{l=1}^L \mu_l+.5)}} F_D^{(2L-1)}(\sum\limits_{l=1}^L \mu_l+0.5, \mu_2-m_2,.,
\end{equation*}
\begin{equation}
 \mu_L-m_L,m_1, ., m_L,\sum\limits_{l=1}^L \mu_l, \frac{a_1-a_2}{a_1},.,\frac{a_1-b_L}{a_1}) 
\label{eqn:g1_kappamu1}
\end{equation}
\else
\begin{equation}
\small
 \frac{G_1 \Gamma(\sum\limits_{l=1}^L \mu_l+0.5 ) F_D^{(2L-1)}(\sum\limits_{l=1}^L \mu_l+0.5, \mu_2-m_2,., \mu_L-m_L,m_1, ., m_L,\sum\limits_{l=1}^L \mu_l, \frac{a_1-a_2}{a_1},.,\frac{a_1-b_L}{a_1})}{ a_1^{(\sum\limits_{l=1}^L \mu_l+.5)}}
\label{eqn:g1_kappamu1}
\end{equation}
\fi

From equations \eqref{eqn:evm_kappamu1}, (\ref{eqn:g2_kappamu1}) and (\ref{eqn:g1_kappamu1}) and using $\theta_{1l}=\frac{1}{a_l}$, $\theta_{2l}=\frac{1}{b_l}$, we get the result.

.\section{}
\section*{Proof of Corollary \ref{cor1}}
\label{app:cor1}
Here we derive the EVM of an interferer limited system when the desired signal experiences $\kappa$-$\mu$ shadowed fading of unit mean power with parameters ($\kappa$, $\mu$, $m$) and interferers experience independent identical $\kappa$-$\mu$ shadowed fading of unit mean power with parameters ($\kappa_I$, $\mu_I$, $m_I$).
As the interferer fading is identical, substitute $\theta_{1l}=\theta_{1I}$, $\theta_{2l}=\theta_{2I}$, $\mu_l=\mu_I,$ $\kappa_l=\kappa_I$, $m_l=m_I$ in (\ref{eqn:evmkappamushadowedint}). Let  
$a_I$=$\frac{1}{\theta_{1I}},$ $b_I=\frac{1}{\theta_{2I}}$. So the Lauricella function in  \eqref{eqn:evmkappamushadowedint} becomes
\ifCLASSOPTIONtwocolumn
\begin{equation*}
F_D^{(2L-1)}(L \mu _I +0.5, \mu _I-m _I,.., \mu _I -m _I, m _I, .., m _I,L \mu _I , 
\end{equation*}
\begin{equation}
 0,..,0,\frac{a_I-b_I}{a_I},..,\frac{a_I-b_I}{a_I}).
 \label{eqn:FD2L-1}
\end{equation}
\else
\begin{equation}
F_D^{(2L-1)}(L \mu _I +0.5, \mu _I-m _I,.., \mu _I -m _I, m _I, .., m _I,L \mu _I , 0,..,0,\frac{a_I-b_I}{a_I},..,\frac{a_I-b_I}{a_I}) .
\label{eqn:FD2L-1}
\end{equation}
\fi
If one of the  parameters of $F_D^{L}$ goes to zero, then it becomes $F_D^{L-1}$ with appropriate parameters. So (\ref{eqn:FD2L-1}) becomes
\begin{equation*}
F_D^L (L \mu _I +0.5, m_I,..,m_I, L \mu _I, \frac{a_I-b_I}
{a_I},..,\frac{a_I-b_I}{a_I}).
\end{equation*}
Now by using the identity \[F_D^{N}(d,b_1,b_2,..b_N,e,x,x,...,x)={}_2 F_1(d,\sum\limits_{i=1}^N b_i, e,x),\] and  $\frac{a_I-b_I}{a_I}=\frac{\mu _I \kappa _I}{\mu _I \kappa _I+m_I}$  we get 
\ifCLASSOPTIONtwocolumn
\begin{equation*}
F_D^L (L \mu _I +0.5, m_I,..,m_I, L \mu _I, \frac{a_I-b_I}
{a_I},..,\frac{a_I-b_I}{a_I}) =  
\end{equation*}
\begin{equation}
{}_2 F_1(L \mu _I +0.5, L m_I, L \mu _I, \frac{\mu _I \kappa _I}{\mu _I \kappa _I + m_I})  
\label{eqn:FD}
\end{equation}
\else
\begin{equation}
\small
F_D^L (L \mu _I +0.5, m_I,..,m_I, L \mu _I, \frac{a_I-b_I}
{a_I},..,\frac{a_I-b_I}{a_I}) = {}_2 F_1(L \mu _I +0.5, L m_I, L \mu _I, \frac{\mu _I \kappa _I}{\mu _I \kappa _I + m_I})  
\label{eqn:FD}
\end{equation}
\fi

Substituting (\ref{eqn:FD}) in (\ref{eqn:evmkappamushadowedint}) and using $\mu_l=\mu _I$, $\kappa_l=\kappa _I$, $m_l=m _I,$    $\forall l=1,..L,$ EVM is obtained.

\section{}
\section*{Proof of Corollary \ref{cor:kappamu}}
\label{app:kappamu}

Here we derive the EVM of an interference limited system when the desired signal experiences unit mean power $\kappa$-$\mu$ fading of parameters ($\kappa$, $\mu$) and the interferers experience identical and independent unit mean power $\kappa$-$\mu$ fading of parameters ($\kappa_I$, $\mu_I$).
From Table I we know that $\kappa$-$\mu$ fading is a special case of $\kappa$-$\mu$ shadowed fading that  can be obtained when $m \rightarrow \infty$ and $m_I$ $\rightarrow \infty$.  So by letting $m$ and $m_I$ approach infinity in (\ref{eqn:evmkappamushadowedintiid}) EVM can be obtained.
 
 So in (\ref{eqn:evmkappamushadowedintiid})  we have to evaluate the term  
\ifCLASSOPTIONtwocolumn
\begin{equation*}
 \left(\frac{m}{m+ \mu \kappa} \right)^{m } {}_2F_1 \left(\mu -.5,m,\mu,\frac{\mu \kappa}{m+ \mu \kappa} \right) 
\end{equation*}
\begin{equation*}
\left(\frac{m_I}{m_I+ \mu_I \kappa_I} \right)^{L m_I } {}_2F_1 \left(L \mu_I +.5, L m_I, L \mu_I , \frac{\mu_I \kappa_I}{\mu_I \kappa_I+m_I} \right),
\end{equation*}
\else
\begin{equation*}
\textstyle
\left(\frac{m_I}{m_I+ \mu_I \kappa_I} \right)^{L m_I } \left(\frac{m}{m+ \mu \kappa} \right)^{m } {}_2F_1 \left(\mu -.5,m,\mu,\frac{\mu \kappa}{m+ \mu \kappa} \right)   {}_2F_1 \left(L \mu_I +.5, L m_I, L \mu_I , \frac{\mu_I \kappa_I}{\mu_I \kappa_I+m_I} \right),
\end{equation*}
\fi
as $m$ and $m_I$ approach infinity.  We will evaluate the effect of $m \rightarrow \infty$ first and similar result is obtained when $m_I \rightarrow \infty$.   Using the series definition of Hypergeometric function \ie,  ${}_p F_q(a_1,...a_p, b_1,...b_q,x) = \sum\limits_{l=0}^{\infty} \frac{(a_1)_l... (a_p)_l}{(b_1)_l....(b_q)_l} \frac{x^l}{l !},$ where $y_l=y(y+1)...(y+l-1),$ we get
\ifCLASSOPTIONtwocolumn
\begin{equation*}
 \frac{ {}_2F_1 \left(\mu -0.5, m, \mu, \frac{\kappa \mu}{m+\kappa \mu} \right) }{ \left(\frac{m+\kappa \mu}{m} \right)^{m}} =   \frac{\sum\limits_{l=0}^{\infty} \frac{(m)_l (\mu-.5)_l}{(\mu)_l} \frac{(\frac{\kappa \mu}{m+\kappa \mu})^l}{l !}}{ \left(\frac{m+\kappa \mu}{m} \right)^{m}}.
 \end{equation*}
\else
\begin{equation*}
 \left(\frac{m}{m+ \kappa \mu} \right)^m {}_2F_1 \left(\mu -0.5, m, \mu, \frac{\kappa \mu}{m+\kappa \mu} \right) = \left(\frac{m}{m+ \kappa \mu} \right)^m  \sum\limits_{l=0}^{\infty} \frac{(m)_l (\mu-.5)_l}{(\mu)_l} \frac{(\frac{\kappa \mu}{m+\kappa \mu})^l}{l !}.
 \end{equation*}
 \fi

$(m)_l$ can also be written as $m^l(1+\frac{1}{m})....(1+\frac{l-1}{m})$.
So taking   $(\frac{m}{m+ \kappa \mu})^m$  within the summation, each term in the summation \ie, for a particular $l$ is 
\begin{equation*}
 \left(\frac{m}{m+\kappa \mu} \right)^{m+l} \frac{(\mu-0.5)_l (\kappa \mu)^l}{l! (\mu)_l} \left(1+\frac{1}{m} \right)....\left(1+\frac{l-1}{m}\right). 
 \end{equation*}

The above equation when $m \rightarrow \infty$ is 
\begin{equation*}
e^{-\kappa \mu} \frac{(\mu-0.5)_l (\kappa \mu)^l}{l! (\mu)_l}. 
\end{equation*}

So when  $m \rightarrow \infty$, 
\ifCLASSOPTIONtwocolumn
\begin{equation*}
 \frac{{}_2F_1 \left(m,\mu -0.5, \mu, \frac{\kappa \mu}{m+\kappa \mu} \right)}{\left(\frac{m+\kappa \mu}{m} \right)^m} =  e^{-\kappa \mu} \sum\limits_{l=0}^{\infty} \frac{(\mu-0.5)_l (\kappa \mu)^l}{l! (\mu)_l} .
\end{equation*}
\else
\begin{equation*}
\left(\frac{m}{m+ \kappa \mu} \right)^m {}_2F_1 \left(m,\mu -0.5, \mu, \frac{\kappa \mu}{m+\kappa \mu} \right) = e^{-\kappa \mu} \sum\limits_{l=0}^{\infty} \frac{(\mu-0.5)_l (\kappa \mu)^l}{l! (\mu)_l} .
\end{equation*}
\fi

From the definition of Hypergeometric function,  $\sum\limits_{l=0}^{\infty} \frac{(\mu-0.5)_l (\kappa \mu)^l}{l! (\mu)_l} ={}_1 F_1(\mu-0.5, \mu, \kappa \mu).$ Hence as $m \rightarrow \infty$,

\begin{equation}
 \frac{{}_2F_1 \left(\mu -0.5, m, \mu, \frac{\kappa \mu}{m+\kappa \mu} \right)}{\left(\frac{m+ \kappa \mu}{m} \right)^m} =  \frac{{}_1 F_1(\mu-0.5, \mu, \kappa \mu)}{e^{\kappa \mu}}.   
\label{eqn:limit1}
\end{equation}

Similarly, as $m_I \rightarrow \infty$,
\ifCLASSOPTIONtwocolumn
\begin{equation*}
\left(\frac{m_I}{m_I+ \kappa_I \mu_I} \right)^{m_I L} {}_2F_1 \left(L \mu_I +0.5, L m_I, L \mu_I, \frac{\kappa_I \mu_I}{m_I+\kappa_I \mu_I} \right) 
\end{equation*}
\begin{equation}
=e^{-\kappa_I \mu_I L} {}_1 F_1(L \mu_I +0.5, L \mu_I, \kappa_I \mu_I L).   
\label{eqn:limit2}
\end{equation}
\else
\begin{equation}
\textstyle
\left(\frac{m_I}{m_I+ \kappa_I \mu_I} \right)^{m_I L} {}_2F_1 \left(L \mu_I +0.5, L m_I, L \mu_I, \frac{\kappa_I \mu_I}{m_I+\kappa_I \mu_I} \right) = \frac{{}_1 F_1(L \mu_I +0.5, L \mu_I, \kappa_I \mu_I L)}{ e^{\kappa_I \mu_I L}}.   
\label{eqn:limit2}
\end{equation}
\fi
Substituting (\ref{eqn:limit1}) and (\ref{eqn:limit2}) in (\ref{eqn:evmkappamushadowedintiid}), and by using the transformation $e^{-z} {}_1F_1(b-a,b,z)={}_1F_1(a,b,-z)$, EVM is obtained.

\section{}
\section*{Proof of Theorem \ref{thm2}}
\label{app:theorem2}
Here we derive the EVM of a system where AWGN  is also present with desired signal experiencing unit mean power $\kappa$-$\mu$ shadowed fading of parameters ($\kappa$, $\mu$ , $m$) and the interferers experiencing identical and independent unit mean power Nakagami-$m$ fading. 
EVM in the presence of AWGN given in (\ref{eqn:evm}) is separable and hence $\int\limits_{0}^{\infty} \int\limits_{0}^{\infty} \sqrt{\frac{g_I + \sigma^2}{g_d}} f_{g_I} f_{g_d} \d g_I \d g_d $ can be expressed as
\begin{equation}
\int\limits_{0}^{\infty} \sqrt{g_I + \sigma^2} f_{g_I} \d g_I  \int\limits_{0}^{\infty} \sqrt{\frac{1}{g_d}} f_{g_d} \d g_d.
\label{eqn:evmdefnnoise}
\end{equation}

Recall that $\int\limits_{0}^{\infty} \sqrt{\frac{1}{g_d}} f_{g_d} \d g_d$ is derived in \eqref{eqn:g2_kappamu1} when the desired signal is $\kappa$-$\mu$ shadowed and it is given as
\begin{equation}
  \frac{ \sqrt{\mu(1+\kappa)} \Gamma(\mu-0.5) {}_2F_1 \left(\mu -\frac{1}{2},m,\mu,\frac{\mu \kappa}{m+ \mu \kappa} \right)}{\Gamma(\mu) \left(\frac{m+\mu \kappa}{m} \right)^m}. 
\label{eqn:g2_kappamu2}
\end{equation}

Interferer powers are i.i.d Gamma distributed of shape and scale parameters $m_I$ and $\frac{1}{m_I}$. So $g_I$ which is a sum of $L$ interferer powers, is also Gamma distributed of shape and scale parameters $m_I L$ and $\frac{1}{m_I}$. Hence $\int\limits_{0}^{\infty} \sqrt{g_I + \sigma^2} f_{g_I} \d g_I $ can be expressed as 

\begin{equation*}
  \int\limits_{0}^{\infty} \sqrt{g_I + \sigma^2} \left(\frac{g_I^{m_I L-1} e^{-g_I m_I } m_I^{m_I L} }{\Gamma(m_I L) } \right) \d g_I. 
 \end{equation*}

Substituting $g_I=t\sigma^2$, $\int\limits_{0}^{\infty} \sqrt{g_I + \sigma^2} f_{g_I} \d g_I $ becomes

\begin{equation*}
\sigma \int\limits_{0}^{\infty} \sqrt{t+1} \left(\frac{(t \sigma^2)^{m_I L-1} e^{-t \sigma^2 m_I } m_I^{m_I L} }{\Gamma(m_I L) } \right) \sigma^2 \d t. 
 \end{equation*}
 
 As the Tricomi confluent hypergeometric function is defined as \[U(a,b,z)= \frac{1}{\Gamma(a)}\int\limits_{0}^{\infty} e^{-z t} t^{a-1} (t+1)^{-a+b-1} \d t,\]
 
\begin{equation}
\int\limits_{0}^{\infty} \sqrt{g_I + \sigma^2} f_{g_I} \d g_I=   \frac{\sigma U \left(m_I L, \frac{3}{2} + m_I L , \sigma^2 m_I\right)}{\left(\sigma^2 m_I \right)^{-m_I L}}.
\label{eqn:g1_kappamu2}
\end{equation}

From equations (\ref{eqn:evmdefnnoise}), (\ref{eqn:g2_kappamu2}) and (\ref{eqn:g1_kappamu2}) we get the result.

\end{appendices}

\bibliographystyle{IEEEtran}
\bibliography{References1}

\end{document}

%% file: notation.tex
\def\E{\mathbb{E}}

\def\ie{{\em i.e.}}

\def\EVM{\mathtt{EVM}}
\def\evm{\mathtt{EVM}}

\def\d{\mathrm{d}}

\def\pdf{\mathtt{PDF}}

{}
  \newtheorem{thm}{Theorem}
  \newtheorem{theorem}{Theorem}
  
  \newtheorem{cor}[thm]{Corollary}